\documentclass[lettersize,journal]{IEEEtran}

\usepackage{tikz}
\usetikzlibrary{positioning}
\usepackage{amsmath}

\usepackage{mathtools}         

\usepackage{amssymb,amsfonts}  
\usepackage{amsthm}
\allowdisplaybreaks
\usepackage{graphicx}
\graphicspath{{figures/}}

\usepackage{booktabs}
\usepackage{multirow}
\usepackage{makecell}
\usepackage{array}
\usepackage{siunitx}
\usepackage{algorithm}
\usepackage{algpseudocode}     

\usepackage{stmaryrd}
\SetSymbolFont{stmry}{bold}{U}{stmry}{m}{n}
\usepackage{xcolor}
\usepackage{balance}
\usepackage{capt-of}

\usepackage{cite}
\usepackage{url}

\theoremstyle{plain}
\newtheorem{theorem}{Theorem}
\newtheorem{proposition}{Proposition}

\theoremstyle{definition}

\newcommand{\system}{EncFormer}

\usepackage[colorlinks=true,linkcolor=blue,citecolor=blue,urlcolor=blue]{hyperref}

\title{\system: Secure and Efficient Transformer Inference over Encrypted Data}
\author{Yufan~Zhu$^{*}$, Chao~Jin$^{\dagger}$, Khin~Mi~Mi~Aung$^{\dagger}$, and Xiaokui~Xiao$^{*}$%
\thanks{$^{*}$National University of Singapore, Singapore.}%
\thanks{$^{\dagger}$Agency for Science, Technology and Research (A*STAR), Singapore.}}
\begin{document}
\maketitle
\begin{abstract}
Transformer inference in machine-learning-as-a-service (MLaaS) raises privacy concerns for sensitive user inputs. Prior secure solutions that combine fully homomorphic encryption (FHE) and secure multiparty computation (MPC) are bottlenecked by inefficient FHE kernels, communication-heavy MPC protocols, and expensive FHE--MPC conversions. We present \system{}, a two-party private Transformer inference framework that introduces Stage Compatible Patterns so that FHE kernels compose efficiently, reducing repacking and conversions. \system{} also provides a cost analysis model built around a minimal-conversion baseline, enabling principled selection of FHE--MPC boundaries. To further reduce communication, \system{} proposes a secure complex CKKS--MPC conversion protocol and designs communication-efficient MPC protocols for nonlinearities. With GPU optimizations, evaluations on GPT- and BERT-style models show that \system{} achieves 1.4$\times$--30.4$\times$ lower online MPC communication and 1.3$\times$--9.8$\times$ lower end-to-end latency against prior hybrid FHE--MPC systems, and 1.9$\times$--3.5$\times$ lower end-to-end latency on BERT-base than FHE-only pipelines under a matched backend, while maintaining near-plaintext accuracy on selected GLUE tasks.

\end{abstract}

\section{Introduction}

Transformer models such as BERT~\cite{devlin2019bert} and GPT~\cite{radford2019language} are increasingly deployed through machine-learning-as-a-service platforms, where user inputs are processed by an untrusted model provider. Sensitive domains such as clinical NLP~\cite{huang2019clinicalbert}, legal text~\cite{chalkidis2020legal}, biomedical mining~\cite{lee2019biobert}, and financial compliance~\cite{yang2023finchainbert} amplify the risk, as both inputs and predictions may be confidential, and attacks such as membership inference~\cite{shokri2017membership} can extract private information from model access alone. Private inference based on fully homomorphic encryption (FHE) and secure multiparty computation (MPC) is a natural fit, but current Transformer pipelines remain expensive because their performance depends not only on the cryptographic primitive but also on how ciphertext layouts, conversion boundaries, and nonlinear protocols interact across the pipeline~\cite{pang2024bolt,kei2025shaft,xu2025blb,moon2025thor,park2025powerformer}.

Prior work on private Transformer inference follows three broad directions. FHE-only systems eliminate online interaction but pay heavily for deep ciphertext computation and, in many cases, bootstrapping~\cite{chen2022thex,moon2025thor,park2025powerformer}. MPC-only systems preserve accurate nonlinear computation but are dominated by repeated interaction on wide activations~\cite{luo2024secformer,kei2025shaft}. Hybrid FHE--MPC systems balance these trade-offs by evaluating wide linear algebra under FHE and delegating nonlinear blocks to MPC, but they inherit a second design problem: every benefit from moving work across the boundary depends on whether the resulting packing can be consumed by the next stage and whether the added conversion cost is justified~\cite{hao2022iron,pang2024bolt,lu2025bumblebee,xu2025blb}.

\subsection{Challenges in Hybrid Private Transformer Inference}

For hybrid pipelines, the main bottlenecks are no longer isolated to a single primitive. First, FHE kernel efficiency depends on packing compatibility across adjacent stages, because a locally efficient kernel can still lose its gain if its output must be repacked before the next FHE operation. Second, FHE--MPC conversion is expensive enough that boundary placement cannot be chosen by intuition alone: using a richer boundary representation may reduce MPC rounds, but only if the extra ciphertext payload and CKKS work are repaid by downstream savings. Third, secure conversion protocols in prior hybrid systems under-utilize CKKS complex slots by exporting only the real part, which increases boundary traffic even when the surrounding pipeline already uses complex packing.

These observations suggest that private Transformer inference should be framed as a secure systems co-design problem rather than as a collection of independent optimizations. The system must expose a consistent cross-stage packing discipline, a secure and communication-efficient conversion protocol, and a boundary decision rule that explains when a non-minimal conversion layout is worthwhile.

\subsection{Our Contribution}

We present \system{}, a two-party hybrid FHE--MPC framework for private Transformer inference in the semi-honest model. The client holds the private input and CKKS secret key, the server holds the model, and the full protocol composes FHE kernels, secure conversion boundaries, and MPC nonlinear blocks. Our main thesis is that hybrid inference becomes more predictable when these components are co-designed around stage compatibility and boundary cost.

We make three integrated contributions.
\begin{itemize}
\item \textbf{Packing-co-designed FHE kernels.}
We introduce Stage Compatible Patterns (SCP), which restrict the pipeline to three canonical packing formats and treat packing as a cross-stage contract. Guided by SCP, \system{} designs a shared plaintext--ciphertext projection kernel and folded-diagonal ciphertext--ciphertext attention kernels whose outputs are directly consumable by downstream stages, avoiding ciphertext remaps between adjacent FHE kernels.
\item \textbf{Boundary co-design and secure conversion.}
We introduce a secure complex CKKS--MPC conversion protocol that uses both the real and imaginary channels to carry two real vectors per ciphertext. While Powerformer uses complex numbers to fuse generic blockwise matrix permutations in its attention kernels, \system{} instead exploits the diagonal and token-shift structure of attention to produce specialized packed layouts (folded-diagonal, head-major) that chain across stages via SCP. \system{} also uses complex conversion at FHE--MPC boundaries to halve the conversion payload. We pair this with a calibrated decision rule for boundary placement under the validated PhantomFHE setting.
\item \textbf{End-to-end secure inference pipeline.}
We implement \system{} on GPUs for GPT2-base, BERT-base, and BERT-large, combining FHE linear kernels with communication-efficient MPC protocols for softmax, layer normalization, and GELU. Evaluations show that \system{} reduces online communication by 30.4$\times$ vs.\ BOLT, 2.6$\times$ vs.\ BumbleBee, and 1.4$\times$ vs.\ BLB, while improving end-to-end latency by 9.8$\times$, 2.1$\times$, and 1.3$\times$ on average, respectively. \system{} also outperforms FHE-only pipelines on BERT-base by 3.5$\times$ over THOR and 1.9$\times$ over Powerformer under a matched backend. On BERT-base, encrypted inference maintains near-plaintext accuracy on selected GLUE tasks (SST-2, MRPC, RTE). Section~\ref{sec:prelim_model} gives the pipeline security proof; kernel derivations, backend-specific comparisons, and supporting analytical details appear in the appendix at the end of this paper.
\end{itemize}

\subsection{System Overview}

Figure~\ref{fig:block-structure} summarizes one \system{} encoder block. The server first applies a shared plaintext--ciphertext projection kernel to produce $Q$, $K$, and $V$ in layouts chosen for their immediate consumers: $Q$ and $K$ are emitted in the score-friendly segment order, while $V$ is emitted in head-major order. The score kernel exports minimal folded-diagonal ciphertexts across the CKKS--MPC boundary, where the MPC softmax block operates on secret shares and returns weights in the same boundary format. The value kernel consumes these weights directly, and the output projection absorbs the remaining layout adaptation through plaintext weight pre-permutation rather than ciphertext repacking. The feed-forward projections follow the same principle, while MPC is reserved for layer normalization and GELU. This organization is the structural basis for SCP, complex conversion, and the boundary cost analysis developed in the following sections.

\begin{figure}[t]
  \centering
  \includegraphics[width=\linewidth]{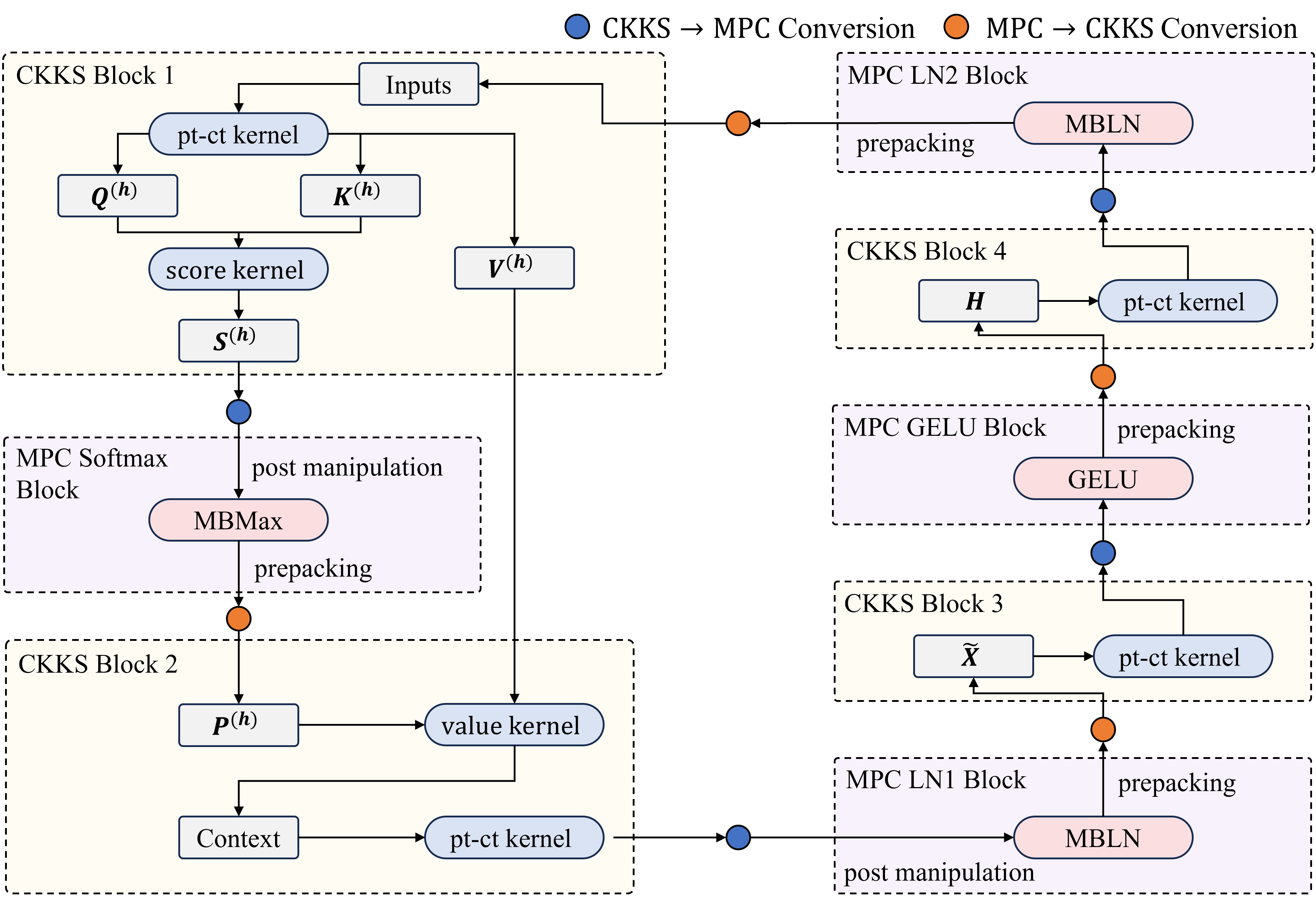}
  \caption{\system{} encoder block showing FHE kernels, MPC nonlinear blocks, and CKKS--MPC conversion boundaries.}
  \label{fig:block-structure}
\end{figure}

\section{Background and Threat Model}
\label{sec:preliminaries}

\subsection{Cryptographic Preliminaries}
\label{sec:prelim_crypto}

\noindent\textbf{Fully Homomorphic Encryption.}
A public-key FHE scheme provides algorithms
$\mathsf{KeyGen}\!\to\!(\mathsf{pk},\mathsf{sk},\mathsf{evk})$,
$\mathsf{Enc}_{\mathsf{pk}}(\cdot)$,
$\mathsf{Dec}_{\mathsf{sk}}(\cdot)$, and
$\mathsf{Eval}_{\mathsf{evk}}(f,\cdot)$ such that, for an arithmetic circuit~$f$,
\[
\mathsf{Dec}_{\mathsf{sk}}\!\bigl(\mathsf{Eval}_{\mathsf{evk}}(f,\mathsf{Enc}_{\mathsf{pk}}(x))\bigr)
=
\begin{cases}
f(x) \bmod t,\\
f(x)+\varepsilon.
\end{cases}
\]
The first line corresponds to exact schemes such as BFV and BGV. The second corresponds to approximate schemes such as CKKS, where $\varepsilon$ captures approximation from encoding, decoding, and noise growth~\cite{fan2012somewhat,brakerski2012leveled,cheon2017ckks}. We use a leveled CKKS instantiation without bootstrapping, implemented on the GPU-accelerated libraries PhantomFHE~\cite{yang2024phantom} and Liberate.FHE~\cite{LiberateFHE}. We refer to~\cite{SEAL41,zheng2023primer} for broader FHE background. The algorithm \(\mathsf{KeyGen}\) samples a public key for encryption, a secret key for decryption, and evaluation keys derived from the secret key. In RLWE-based FHE, \(\mathsf{evk}\) typically includes relinearization keys to manage ciphertext growth after multiplication and Galois keys that enable SIMD slot permutations such as rotations and conjugation. The algorithm \(\mathsf{Enc}_{\mathsf{pk}}\) encodes and encrypts a plaintext into a ciphertext, \(\mathsf{Eval}_{\mathsf{evk}}\) applies additions, multiplications, and slot permutations directly on ciphertexts using only public evaluation material, and \(\mathsf{Dec}_{\mathsf{sk}}\) recovers the plaintext result~\cite{gentry2009a,brakerski2012leveled,halevi2014algorithms}.

\noindent\textbf{Secure Multiparty Computation.}
Secure multiparty computation allows parties to jointly evaluate a function over private inputs while revealing only the prescribed outputs. We use two-party arithmetic secret sharing over $\mathbb{Z}_{2^\ell}$ and represent real values in fixed point with scale $s=2^F$, where $F$ is the fractional bit-width and $\ell$ is the ring bit-width. Linear maps are computed locally on shares. To multiply shared values, we use Beaver triples $(a, b, c)$ with $c=ab$~\cite{beaver1992efficient}. Given shared $x$ and shared $y$, the parties open $\alpha=x-a$ and $\beta=y-b$, then obtain a sharing of the product by setting $z=xy$ and computing
\[
z^{\mathrm{sh}} = c^{\mathrm{sh}} + \alpha\,b^{\mathrm{sh}} + \beta\,a^{\mathrm{sh}} + \alpha\beta,
\]
where the public term $\alpha\beta$ is incorporated into one party's share. For fixed-point arithmetic, multiplying two $s$-scaled values produces a value at scale $s^2$, so we apply an interactive secure truncation protocol that right shifts by $F$ bits to return to scale $s$. This yields an online phase with one interaction round per multiplication depth and communication proportional to circuit size, as in SPDZ-style protocols~\cite{Damgard2012SPDZ}. In the semi-honest model, security is captured by the real/ideal paradigm~\cite{yao1986how,goldreich1987how}.

In \system{}, we use a hybrid FHE--MPC setting. CKKS provides SIMD-style approximate arithmetic on vectors packed into ciphertext slots, supporting ciphertext addition, ciphertext/plaintext multiplication, ciphertext/ciphertext multiplication, rotations, and conjugation~\cite{cheon2017ckks}. We adopt CKKS because Transformer inference is dominated by real-valued linear algebra and already operates in finite precision, so it tolerates small approximation error. This choice is consistent with prior hybrid FHE--MPC Transformer systems~\cite{hao2022iron,pang2024bolt,xu2025blb,zhang2025nexus}. CKKS's scale-and-rescale workflow matches the way Transformer layers accumulate depth through chained linear maps, and its SIMD packing over complex slots enables batching and reduces ciphertext count for matrix workloads. We exploit complex packing inside our FHE kernels with SCP and at FHE--MPC conversion boundaries. For nonlinear layers, we use the MPC building blocks above. The resulting latency is sensitive to the number of rounds and the amount of data crossing each MPC boundary, which is why boundary design matters as much as primitive selection in a hybrid pipeline.

\subsection{CKKS Parameters}
\label{sec:prelim_ckks}

We instantiate FHE with the RNS--CKKS scheme of Cheon et al.~\cite{cheon2017ckks}, which supports approximate arithmetic on vectors packed into ciphertext slots. CKKS operates over the cyclotomic ring $R_Q=\mathbb{Z}_Q[X]/(X^N+1)$ for power-of-two $N$, and packs $n=N/2$ complex slots via an approximate encoding map. Figure~\ref{fig:rns-ckks} illustrates the residue number system (RNS) representation and how computation proceeds across the modulus chain. In its RNS form, the ciphertext modulus is represented as a product of machine-word primes $Q=\prod_{i=0}^{L} q_i$, and polynomial coefficients are stored and processed as residues modulo each $q_i$, enabling efficient NTT-based arithmetic and scale management across the chain.

\begin{figure}[t]
  \centering
  \includegraphics[width=0.8\linewidth]{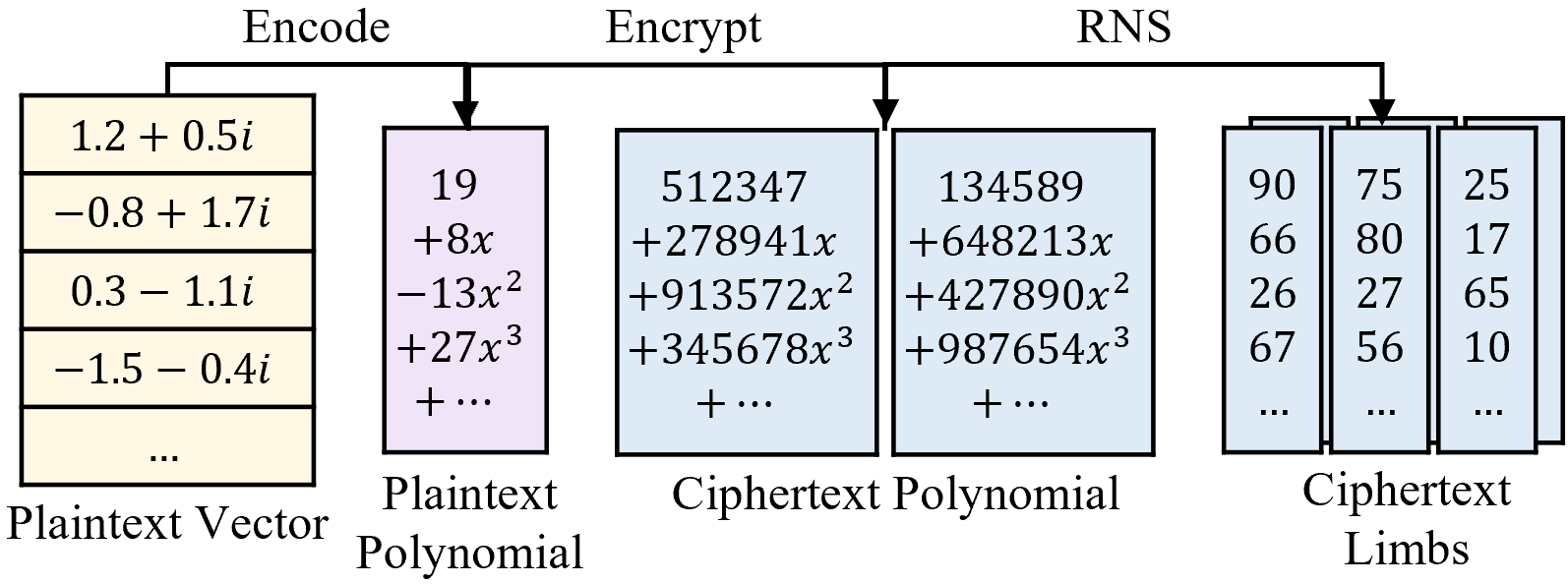}
  \caption{Illustration of RNS--CKKS, showing the residue representation and the modulus-chain progression under rescaling.}
  \label{fig:rns-ckks}
\end{figure}

Practical deployments fix a modulus chain \(Q=\prod_{i=0}^{L} q_i\) and a target scale \(\Delta\) sized for multiplicative depth \(D\le L\). Let \(Q_i=\prod_{j=0}^{i} q_j\) denote the remaining ciphertext modulus at level \(i\). A ciphertext--ciphertext multiplication inflates the scale by roughly a factor of \(\Delta\) and is followed by \(\mathsf{rescale}\), which maps the ciphertext modulus from \(Q_i\) to \(Q_{i-1}\) by dividing by \(q_i\) and restores the scale to approximately \(\Delta\). After \(u\) multiplication--rescale steps, the remaining modulus is \(Q_{L-u}\). Ciphertext size and remaining multiplicative depth both depend on the active level of this chain, so boundary-conversion cost depends not only on tensor shape but also on the CKKS parameters at the conversion point. Our cost model therefore treats ciphertext size as a function of the ring degree and the remaining RNS limbs.

\subsection{Threat Model}
\label{sec:prelim_model}

We consider the standard two-party private inference setting with a client $P_0$ and a server $P_1$. The client provides the private input and learns the inference output; the server provides the model. We assume semi-honest adversaries: both parties follow the protocol but may try to infer additional information from the transcript. As in prior hybrid systems, the model architecture is known to both parties~\cite{huang2022cheetah,juvekar2018gazelle,lu2025bumblebee,pang2024bolt,xu2025blb}.

\noindent\textbf{Key Distribution.}
Following prior hybrid systems~\cite{juvekar2018gazelle,pang2024bolt,xu2025blb,lu2025bumblebee}, the client generates the CKKS key material $(\mathsf{pk},\mathsf{sk},\mathsf{evk})$ and sends $(\mathsf{pk},\mathsf{evk})$ to the server, while retaining the secret key locally. This is the standard key-distribution pattern in two-party hybrid inference~\cite{juvekar2018gazelle,pang2024bolt,xu2025blb}: the client holds the secret key so that decryption occurs only on the client side, and the server never observes plaintext activations. During CKKS-to-MPC conversion, the client decrypts only one-time-padded ciphertexts, so the plaintext observed by either party is statistically or computationally protected by the underlying conversion protocol.

The inference protocol is organized as a sequence of components $\Pi_{\mathsf{Infer}}=\Pi_1\circ\cdots\circ\Pi_K$, where each $\Pi_j$ is a CKKS evaluation block, an MPC block over $\mathbb{Z}_{2^\ell}$, or an FHE--MPC conversion boundary. Every intermediate value is represented either as a CKKS ciphertext under the secret key or as additive secret shares held by the two parties.
\begin{theorem}[Pipeline Security]
  \label{thm:pipeline-security}
  If CKKS is semantically secure and the MPC protocols and complex conversion protocols are secure in the semi-honest threat model, then $\Pi_{\mathsf{Infer}}$ securely realizes the ideal inference functionality in the semi-honest model.
\end{theorem}

\begin{proof}
  Fix a party $P_b$ for $b\in\{0,1\}$.
  Let $\mathsf{view}_b$ denote its real-execution view and $\mathsf{out}$ the prescribed client output.
  We construct a simulator $\mathcal{S}_b$ such that
  $\bigl\{\mathcal{S}_b(\mathsf{in}_b,\mathsf{out}),\mathsf{out}\bigr\}\equiv_c
    \bigl\{\mathsf{view}_b,\mathsf{out}\bigr\}$.
  Each component $\Pi_j$ is secure by assumption, so it admits a simulator $\mathcal{S}^{(j)}_b$ producing a transcript indistinguishable from the real one: for CKKS blocks this follows from semantic security, for MPC blocks from semi-honest MPC security, and for conversion boundaries from Section~\ref{sec:complex-conversion}. The global simulator $\mathcal{S}_b$ composes these component simulators in order. A standard hybrid argument over $\mathsf{H}_0,\ldots,\mathsf{H}_K$ (where $\mathsf{H}_j$ replaces the first $j$ transcripts with simulated ones) yields $\mathsf{H}_0\equiv_c \mathsf{H}_K$, showing that the composed protocol leaks nothing beyond the prescribed output.
\end{proof}

\subsection{Transformer Model}
\label{transformer_models}

The Transformer is a sequence model built around self-attention, replacing recurrence and convolution with token--token interactions that can be computed in parallel~\cite{Vaswani2017Attention}. For a sequence of length \(m\) with hidden size \(d_{\mathrm{model}}\) and activations \(A\in\mathbb{R}^{m\times d_{\mathrm{model}}}\), the attention sublayer forms \(Q=AW_Q\), \(K=AW_K\), and \(V=AW_V\) with \(W_Q,W_K,W_V\in\mathbb{R}^{d_{\mathrm{model}}\times d_{\mathrm{model}}}\). With \(H\) heads of width \(d_h=d_{\mathrm{model}}/H\), write \(Q^{(h)},K^{(h)},V^{(h)}\in\mathbb{R}^{m\times d_h}\) for the \(h\)th head slices. A compact head computation is
\[
\mathrm{Att}^{(h)} \;=\;
\mathrm{softmax}\!\left(\frac{Q^{(h)}(K^{(h)})^\top}{\sqrt{d_h}}\right)V^{(h)}.
\]
The multi-head output concatenates all heads and applies an output projection \(W_O\in\mathbb{R}^{d_{\mathrm{model}}\times d_{\mathrm{model}}}\),
\[
\mathrm{MHA}(A) \;=\; \mathrm{Concat}\!\bigl(\mathrm{Att}^{(0)},\dots,\mathrm{Att}^{(H-1)}\bigr)W_O.
\]
This is followed by a position-wise feed-forward sublayer \(\mathrm{FFN}(x)=\phi(xW_1+b_1)W_2+b_2\) using \(\phi=\mathrm{GELU}\), with residual connections and layer normalization throughout~\cite{Vaswani2017Attention}. Figure~\ref{fig:transformer-arch} shows the reference architecture used in this work.

BERT is an encoder-only Transformer that stacks \(N_L\) identical layers to build bidirectional context~\cite{devlin2019bert}. We evaluate BERT-base with \(N_L=12\), \(d_{\mathrm{model}}=768\), intermediate size \(3072\), and \(H=12\), and BERT-large with \(N_L=24\), \(d_{\mathrm{model}}=1024\), intermediate size \(4096\), and \(H=16\), where both have \(d_h=64\). We also evaluate GPT2-base, a decoder-only Transformer with causal self-attention~\cite{radford2019language}, which matches BERT-base in \(N_L\), \(d_{\mathrm{model}}\), \(H\), and intermediate size. These dimensions determine the packing sizes, boundary tensor shapes, and kernel schedules used in the later sections.

\begin{figure}[t]
  \centering
  \includegraphics[width=\linewidth]{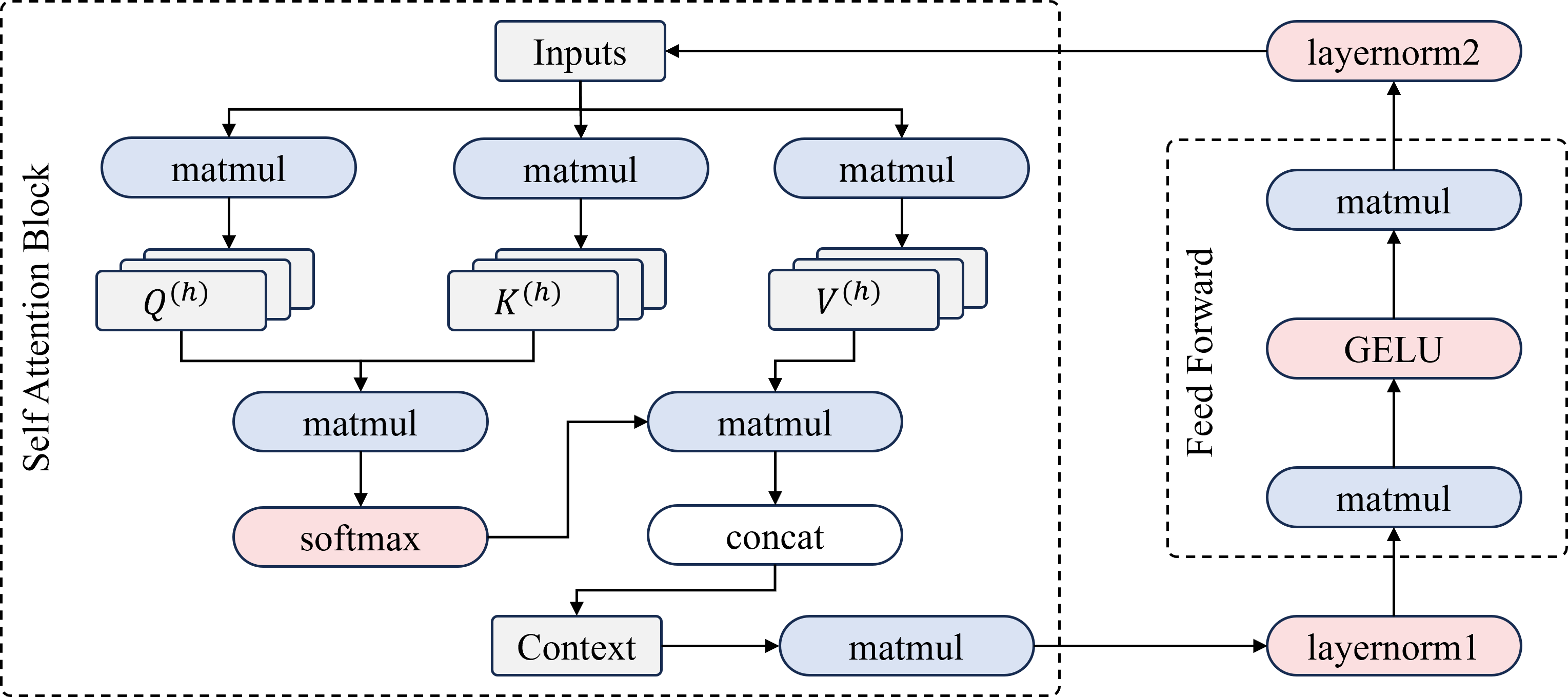}
  \caption{Reference Transformer architecture used in this work, highlighting attention, feed-forward layers, and the output head.}
  \label{fig:transformer-arch}
\end{figure}

\section{Packing-Co-Designed FHE Kernels}
\label{sec:ks-he}

This section presents the FHE side of \system{}. We first define the packing notation and CKKS primitives used by the kernels, then introduce Stage Compatible Patterns (SCP) as the cross-stage layout discipline, and finally describe the shared projection kernel and the ciphertext--ciphertext attention kernels.

A key computational characteristic of the Transformer architecture is the attention block in Section~\ref{transformer_models}, whose runtime is dominated by large linear projections and matrix multiplications. Following prior hybrid Transformer systems~\cite{hao2022iron,pang2024bolt,zhang2025nexus,lu2025bumblebee,xu2025blb}, \system{} evaluates these linear components under FHE.

\subsection{Packing Model and Stage Compatibility}
\label{sec:ks-he:packing}

Let $n$ denote the number of slots in one CKKS ciphertext.
We partition the $n$ slots into $N_{\mathrm{seg}} = n/m$ contiguous segments, each containing $m$ slots, so that $m N_{\mathrm{seg}} = n$.
We write bold lowercase letters for slot vectors in $\mathbb{C}^n$.
For a vector $\mathbf{v}\in\mathbb{C}^n$, we denote its CKKS encryption by $\langle \mathbf{v}\rangle$.
For $\mathbf{v},\mathbf{w}\in\mathbb{C}^n$, we use $\mathbf{v}+\mathbf{w}$ and $\mathbf{v}\odot\mathbf{w}$ for elementwise addition and multiplication, and $\overline{\mathbf{v}}$ for complex conjugation.
Let $\rho(\mathbf{v};r)$ denote the cyclic left shift of $\mathbf{v}$ by $r$ positions, where $r\in\mathbb{Z}$ and indices are taken modulo $n$.
\[
\rho(\mathbf{v};r) \;=\; (v_r, v_{r+1}, \dots, v_{n-1}, v_0, \dots, v_{r-1}).
\]

\begin{table}[t]
\centering
\footnotesize
\setlength{\tabcolsep}{5pt}
\renewcommand{\arraystretch}{1.15}
\caption{Homomorphic operations on CKKS ciphertexts.}
\label{tab:he_ops}
\begin{tabular}{@{}l|ccccc@{}}
\toprule
 primitive & \textsf{add} & \textsf{ptmul} & \textsf{ctmul} & \textsf{rot} & \textsf{conj} \\
\midrule
Inputs
& $\bigl(\langle\mathbf{v}\rangle,\,u\bigr)$
& $\bigl(\langle\mathbf{v}\rangle,\,\mathbf{w}\bigr)$
& $\bigl(\langle\mathbf{v}\rangle,\,\langle\mathbf{w}\rangle\bigr)$
& $\bigl(\langle\mathbf{v}\rangle;\,r\bigr)$
& $\bigl(\langle\mathbf{v}\rangle\bigr)$ \\
Output
& $\langle \mathbf{v}+u\rangle$
& $\langle \mathbf{v}\odot \mathbf{w}\rangle$
& $\langle \mathbf{v}\odot \mathbf{w}\rangle$
& $\langle \rho(\mathbf{v};r)\rangle$
& $\langle \overline{\mathbf{v}}\rangle$ \\
\bottomrule
\end{tabular}
\end{table}

To express segment structure, let $\mathbf{x}\in\mathbb{C}^n$ be a slot vector and index slots by a within-segment coordinate $r\in\{0,\dots,m-1\}$ and a segment coordinate $s\in\{0,\dots,N_{\mathrm{seg}}-1\}$.
Define the reshape operator $\operatorname{mat}:\mathbb{C}^n\to\mathbb{C}^{m\times N_{\mathrm{seg}}}$ by
\[
\bigl(\operatorname{mat}(\mathbf{x})\bigr)_{r,s} \;=\; \mathbf{x}_{s\,m + r},
\]
and let $\operatorname{vec}:\mathbb{C}^{m\times N_{\mathrm{seg}}}\to\mathbb{C}^n$ be its inverse that stacks columns back into a length $n$ vector.
Thus $\operatorname{vec}(\operatorname{mat}(\mathbf{x}))=\mathbf{x}$ and $\operatorname{mat}(\operatorname{vec}(\mathbf{X}))=\mathbf{X}$ for any $\mathbf{X}\in\mathbb{C}^{m\times N_{\mathrm{seg}}}$.

\noindent\textbf{Ciphertext transformation operations.}
Following the transformation notation of Jiang et al.~\cite{jiang2018secure}, for $\mathbf{X}\in\mathbb{C}^{m\times N_{\mathrm{seg}}}$ we define a segment shift $\phi^\Delta$ and an intra-segment shift $\psi^t$ by
\[
\begin{aligned}
\phi^\Delta(\mathbf{X})_{r,s} &= \mathbf{X}_{r,\,(s+\Delta)\bmod N_{\mathrm{seg}}},\\
\psi^t(\mathbf{X})_{r,s} &= \mathbf{X}_{(r+t)\bmod m,\,s}.
\end{aligned}
\]
These shifts induce permutations on slot vectors,
\[
\begin{aligned}
\Phi^\Delta(\mathbf{x}) &= \operatorname{vec}\!\bigl(\phi^\Delta(\operatorname{mat}(\mathbf{x}))\bigr),\\
\Psi^t(\mathbf{x}) &= \operatorname{vec}\!\bigl(\psi^t(\operatorname{mat}(\mathbf{x}))\bigr).
\end{aligned}
\]
On ciphertexts, we apply the corresponding slot permutations and write
$\Phi^\Delta(\langle \mathbf{x}\rangle)=\langle \Phi^\Delta(\mathbf{x})\rangle$ and $\Psi^t(\langle \mathbf{x}\rangle)=\langle \Psi^t(\mathbf{x})\rangle$.
Intuitively, $\Phi^\Delta$ shifts by whole segments and can be realized as $\textsf{rot}(\langle \mathbf{x}\rangle;\Delta m)$, while $\Psi^t$ shifts within each length $m$ segment.
Concrete realizations can be found in \hyperref[app:he-shifts]{Appendix~A.1}.
When only the first $C<N_{\mathrm{seg}}$ segments are active and wrap-around is required within the active region, we use the restricted segment rotation $\Phi_C^\Delta$, also detailed in \hyperref[app:he-shifts]{Appendix~A.1}. When an object $X$ is represented by multiple ciphertext blocks, we write
$\langle X\rangle=\{\langle \mathbf{x}^{(\ell)}\rangle\}_{\ell}$.
We overload $\Phi^\Delta(\cdot)$ and $\Psi^t(\cdot)$ to act elementwise on lists.

\noindent\textbf{Minimal and expanded packings.}
Let $x$ be a real-valued vector with $N(x)$ logical scalar entries, where $N(x)$ excludes repetitions introduced for broadcasting, alignment, or interleaving. Under complex packing, each slot stores two independent real scalars, one in the real part and one in the imaginary part.
The natural lower bound on ciphertext count is
\[
K_{\min}(x) \;=\; \left\lceil \frac{N(x)}{2n} \right\rceil .
\]
A packing of $x$ into $K$ ciphertexts is information-preserving if it stores all $N(x)$ entries without loss.
A packing is \emph{minimal} if it is information-preserving and attains $K = K_{\min}(x)$, and \emph{expanded} if it is information-preserving but uses $K > K_{\min}(x)$.
Figure~\ref{fig:packing_compare} illustrates the distinction on an example.

\begin{figure}[t]
  \centering
  \vspace{-2pt}
  \includegraphics[width=\linewidth]{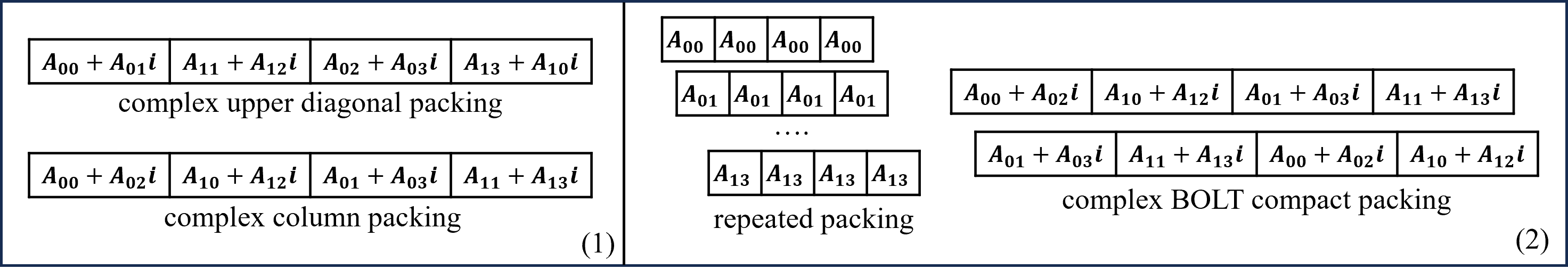}
  \caption{Packing comparison for a $2\times4$ matrix with $n=4$ slots per ciphertext. Minimal packing uses $K_{\min}=1$ ciphertext, whereas expanded packing uses more than $K_{\min}$ ciphertexts.}
  \label{fig:packing_compare}
\end{figure}

\subsubsection{Stage Compatible Patterns}
\label{sec:ks-he:contracts}
A hybrid FHE--MPC pipeline has two sources of avoidable overhead. First, adjacent FHE kernels may use incompatible packings, forcing an expensive ciphertext remap between them. Second, FHE--MPC boundaries may use more ciphertexts than necessary, inflating conversion communication. Stage Compatible Patterns address both problems through three rules applied to a canonical packing family $\mathcal{F}=\{\text{segment-column},\allowbreak\ \text{folded-diagonal},\allowbreak\ \text{head-major}\}$.

Each FHE kernel declares which format it accepts and which it produces. The three SCP rules are:
\begin{enumerate}
  \item \textbf{FHE$\to$FHE.} The producer outputs in the format the consumer expects, so no repacking is needed.
  \item \textbf{MPC$\to$FHE.} Data returning from MPC enters the next FHE kernel in minimal packing.
  \item \textbf{FHE$\to$MPC.} Data leaving an FHE kernel for MPC is exported in minimal packing.
\end{enumerate}
Rule~1 eliminates FHE-side repacking. Rules~2 and~3 minimize the number of ciphertexts at each conversion boundary, directly reducing communication. The quantitative impact of Rules~2 and~3 on boundary cost is analyzed in Section~\ref{sec:complex-conversion}.

\begin{proposition}[Stage Compatible Patterns]
\label{prop:scp-zero-repacking}
Under the kernel interfaces of \system{}, SCP achieves (i)~zero FHE remap stages on all edges, and (ii)~minimal ciphertext count $K_{\min}$ at every FHE--MPC boundary. Violating Rule~1 on an FHE--FHE edge incurs $\Omega(m\log m)$ extra rotations per ciphertext~\cite{halevi2014algorithms}. Violating Rules~2 or~3 increases boundary conversion count beyond $K_{\min}$. The proof is in \hyperref[app:scp-proof]{Appendix~F}.
\end{proposition}

\subsection{Plaintext--Ciphertext Projection Kernel}
\label{sec:ks-he:ctpt}

For plaintext--ciphertext multiplication, we use a shared kernel for linear projections $Y=XW$, which instantiates a CKKS plaintext--ciphertext matrix multiplication under segment-column packing using standard baby-step-giant-step (BSGS) optimization~\cite{xu2024privcirnet,pang2024bolt} and complex-packing optimizations~\cite{park2025powerformer}.

Let $X\in\mathbb{R}^{m\times d_{\mathrm{in}}}$ and $W\in\mathbb{R}^{d_{\mathrm{in}}\times d_{\mathrm{out}}}$.
A CKKS ciphertext has $n$ slots, viewed as $N_{\mathrm{seg}}=n/m$ contiguous segments of length $m$ via
$\bigl(\operatorname{mat}(\mathbf{x})\bigr)_{r,c}=\mathbf{x}[r+cm]$.
We activate the first $C=c_{\mathrm{used}}\le N_{\mathrm{seg}}$ segments so one ciphertext packs $C$ columns.
We group input channels into $G=\lceil d_{\mathrm{in}}/C\rceil$ ciphertexts and output channels into
$B_{\mathrm{out}}=\lceil d_{\mathrm{out}}/C\rceil$ ciphertexts.

For each $g=0,\dots,G-1$, we encrypt $\langle \mathbf{x}^{(g)}\rangle$ such that segment $c\in\{0,\dots,C-1\}$ stores the length-$m$ column $X[:,\,gC+c]$.
For each $b=0,\dots,B_{\mathrm{out}}-1$, the output ciphertext $\langle \mathbf{y}^{(b)}\rangle$ stores columns
$Y[:,\,bC],\dots,Y[:,\,bC+C-1]$ in the same layout.
Projection weights are encoded as CKKS plaintext vectors as detailed in \hyperref[app:he-kernels:proj]{Appendix~A.2}.
We pre-permute plaintext weight columns so each projection emits the order consumed by its downstream kernel directly: $Q$ and $K$ use the score-friendly segment order, while $V$ uses the value-friendly head-major order.

\noindent\textbf{Complex pairs compression.}
We pack two real input groups into one complex ciphertext by setting $U=\lceil G/2\rceil$ and, for each $u\in{0,\dots,U-1}$, defining the complexified input $\langle \tilde{\mathbf{x}}^{(u)}\rangle$ as
\[
\langle \tilde{\mathbf{x}}^{(u)} \rangle
=
\left\langle \mathbf{x}^{(2u)} + i\,\mathbf{x}^{(2u+1)} \right\rangle.
\]
We call this packing operation complexification and call the reverse operation decomplexification.

\noindent\textbf{BSGS optimization.}
Choose $N_1 \mid C$ and set $N_2 = C/N_1$.
Let $\Phi_C$ denote a segment cyclic shift within the first $C$ active segments.
For each $u$ and $q\in\{0,\dots,N_1-1\}$ define $\langle \tilde{\mathbf{x}}^{(u)}_q\rangle=\Phi_C^q(\langle \tilde{\mathbf{x}}^{(u)}\rangle)$.
For each output block $b$ and giant index $p\in\{0,\dots,N_2-1\}$ we accumulate a ciphertext
\[
\langle \tilde{\mathbf{c}}^{(b)}_p\rangle
=
\sum_{q=0}^{N_1-1}\sum_{u=0}^{U-1}
\left(\langle \tilde{\mathbf{x}}^{(u)}_q\rangle \odot \tilde{\mathbf{w}}^{(b)}_{u,p,q}\right),
\]

We then decomplexify each giant-step accumulator by extracting its real part, and combine the $N_2$ partial sums with segment shifts to obtain $\langle \mathbf{y}^{(b)}\rangle$:

\[
\begin{aligned}
\langle \mathbf{c}^{(b)}_p\rangle
&=\tfrac12\Bigl(\langle \tilde{\mathbf{c}}^{(b)}_p\rangle+\textsf{conj}(\langle \tilde{\mathbf{c}}^{(b)}_p\rangle)\Bigr),\\
\langle \mathbf{y}^{(b)}\rangle
&=\sum_{p=0}^{N_2-1}\Phi_C^{pN_1}\!\left(\langle \mathbf{c}^{(b)}_p\rangle\right).
\end{aligned}
\]

\hyperref[app:he-kernels:proj]{Appendix~A.2} also describes an optional shortcut that fuses two projections that share the same input.
The projection kernel satisfies SCP Rule~1 because it outputs in the segment-column packing consumed by the score kernel, and Rule~2 because MPC-produced inputs are imported into CKKS in minimal packing and consumed directly as $\{\langle \mathbf{x}^{(g)}\rangle\}$.

\subsection{Ciphertext--Ciphertext Attention Kernels}
\label{sec:ks-he:ctct}

Transformer attention is head-separable and exhibits a token-shift reuse pattern that we exploit to reduce key-switching operations relative to generic ciphertext--ciphertext matrix multiplication. For each head $h$, we compute the score matrix $S^{(h)}=Q^{(h)}(K^{(h)})^\top\in\mathbb{R}^{m\times m}$ with a score kernel and the attention output $O^{(h)}=P^{(h)}V^{(h)}\in\mathbb{R}^{m\times d_h}$ with a value kernel, where $P^{(h)}$ is the post-softmax attention matrix.

Both kernels reuse the same small sets of segment shifts and intra-segment shifts on each ciphertext block, so we precompute rotation banks to avoid redundant key-switching operations.
For block $\ell$ with input $\langle \mathbf{x}^{(\ell)}\rangle$, let $\mathcal{D},\mathcal{T}\subset\mathbb{Z}$ be the required segment and intra-segment offsets.
We precompute
\[
\mathsf{Bank}_{\Phi}\!\left(\langle \mathbf{x}^{(\ell)}\rangle,\mathcal{D}\right)
=
\left\{\Phi^{\delta}\!\left(\langle \mathbf{x}^{(\ell)}\rangle\right)\right\}_{\delta\in\mathcal{D}},
\]
\[
\mathsf{Bank}_{\Psi}\!\left(\langle \mathbf{x}^{(\ell)}\rangle,\mathcal{T}\right)
=
\left\{\Psi^{t}\!\left(\langle \mathbf{x}^{(\ell)}\rangle\right)\right\}_{t\in\mathcal{T}}.
\]

When only the first $C<N_{\mathrm{seg}}$ segments are active, we use $\Phi_C^\delta$ so that wrap-around stays within the active region.
The inner loops then use bank lookups with $\textsf{ptmul}$, $\textsf{add}$, and $\textsf{ctmul}$, and all rotations occur only during bank construction.

\subsubsection{Score Kernel}
Inputs are $\langle Q\rangle$ and $\langle K\rangle$ in segment-column packing, where
$\langle Q\rangle=\{\langle \mathbf{q}^{(\ell)}\rangle\}_{\ell}$ and
$\langle K\rangle=\{\langle \mathbf{k}^{(\ell)}\rangle\}_{\ell}$ encrypt $Q$ and $K$.
The output is a \emph{folded-diagonal} ciphertext list $\{\langle \mathbf{S}_t\rangle\}_{t=0}^{m/2-1}$. Folded-diagonal packing compresses two attention diagonals into one ciphertext: for each head $h$, define the cyclic diagonal vector at offset $\Delta$ as
\[
s^{h}_{\Delta}[j] = S^{h}_{j,\, (j+\Delta) \bmod m}, \quad j = 0,\dots, m-1.
\]
The offset pair $(t, t + m/2)$ is stored by placing their diagonals into the real and imaginary parts of a single vector:
\[
v^{h}_{t}[j] = s^{h}_{t}[j] + i\, s^{h}_{t + m/2}[j], \quad t = 0,\dots, m/2 - 1.
\]
We assign one length-$m$ segment per head inside each ciphertext, with heads stacked across segments, and place different values of $t$ across ciphertext blocks. Figure~\ref{fig:ccmm-s-stage} illustrates this layout in the kernel procedures.

Choose $\beta\mid m$ and set $g=m/\beta$ with $g$ even. We compute each folded-diagonal pair by reusing a fixed set of intra-segment shifts on $Q$ and $K$.
For each $t\in\{0,\dots,m/2-1\}$, write $t=j(t)\beta+s(t)$ with
$j(t)=\lfloor t/\beta\rfloor$ and $s(t)=t\bmod\beta$.
Pairing $(t,t+m/2)$ matches folded-diagonal packing, so each $t$ is produced by one complex multiply
between a $Q$ shift and a complexified pair of $K$ shifts.

For each ciphertext block $\ell$, we precompute a $Q$ bank over the offsets
\[
\mathcal{T}_Q=\{0,-1,\dots,-(\beta-1)\},
\]
where the negative shift aligns the $Q$ row index with the $K$ diagonal offset $s(t)$ in the inner product.
We also precompute two $K$ banks over the offsets
\[
\begin{aligned}
\mathcal{T}_K &= \{j\beta\mid j=0,\dots,g/2-1\},\\
\mathcal{T}_{K,\mathrm{half}} &= \{m/2+j\beta\mid j=0,\dots,g/2-1\}.
\end{aligned}
\]
This yields
\[
\begin{aligned}
\{\langle \mathbf{q}^{(\ell)}_{s}\rangle\}_{s}
&:=
\mathsf{Bank}_{\Psi}\bigl(\langle \mathbf{q}^{(\ell)}\rangle,\mathcal{T}_Q\bigr),\\
\{\langle \mathbf{k}^{(\ell)}_{\tau}\rangle\}_{\tau}
&:=
\mathsf{Bank}_{\Psi}\bigl(\langle \mathbf{k}^{(\ell)}\rangle,\mathcal{T}_K\cup\mathcal{T}_{K,\mathrm{half}}\bigr),
\end{aligned}
\]
Using these banks, we form one folded-diagonal pair per $t$,
\[
\langle \mathbf{u}_t\rangle
=
\sum_{\ell}
\Bigl(
\langle \mathbf{q}^{(\ell)}_{-s(t)}\rangle
\odot
\bigl(
\langle \mathbf{k}^{(\ell)}_{j(t)\beta}\rangle
+
i\,\langle \mathbf{k}^{(\ell)}_{m/2+j(t)\beta}\rangle
\bigr)
\Bigr).
\]
Finally, we aggregate the $C$ active segments into the first $H$ head segments and align the intra-segment offset.
Let $\{\mathbf{m}_c\}_{c=0}^{C-1}$ be disjoint binary segment masks, where $\mathbf{m}_c$ keeps exactly the $c$-th active segment.
This produces $\langle \mathbf{S}_t\rangle$ in folded-diagonal packing,
\[
\begin{aligned}
\langle \mathbf{S}_t\rangle
&=
\Psi^{s(t)}
\Biggl(
\sum_{c=0}^{C-1}
\Phi^{(c\bmod H)-c}\bigl(\langle \mathbf{u}_t\rangle \odot \mathbf{m}_c\bigr)
\Biggr).
\end{aligned}
\]
When $C\not\equiv 0\pmod H$, different ciphertext blocks start at different head-phase offsets, requiring a block-dependent phase correction before combining. Mask construction and the full phase-correction procedure are given in \hyperref[app:he-kernels:qk]{Appendix~A.3}; the equations above describe the core accumulation without this correction for clarity.

\subsubsection{Value Kernel}
Let $H_{\mathrm{blk}}$ be the number of heads packed in one value ciphertext block and let
$B_V=\lceil H/H_{\mathrm{blk}}\rceil$.
The CKKS inputs are the folded-diagonal softmax ciphertext list
\[
\langle P_{\mathrm{fd}}\rangle=\{\langle \mathbf{p}^{(\ell)}_{\mathrm{fd}}\rangle\}_{\ell=0}^{B_V-1}
\]
and the value ciphertext list
\[
\langle V\rangle=\{\langle \mathbf{v}^{(\ell)}\rangle\}_{\ell=0}^{B_V-1}.
\]
Each value ciphertext $\langle \mathbf{v}^{(\ell)}\rangle$ uses \emph{head-major} packing: segment
$k(\tilde h,u)=\tilde h\,d_h+u$ stores $V^{(\ell,\tilde h)}[:,u]\in\mathbb{R}^m$ for
$\tilde h\in\{0,\dots,H_{\mathrm{blk}}-1\}$ and $u\in\{0,\dots,d_h-1\}$.
Thus, the $d_h$ channels of each local head are contiguous, and different heads occupy disjoint blocks of $d_h$ segments.
The output is the head-major ciphertext list
\[
\langle O_{\mathrm{hm}}\rangle=\{\langle \mathbf{o}^{(\ell)}_{\mathrm{hm}}\rangle\}_{\ell=0}^{B_V-1}
\]
which encodes $O=PV\in\mathbb{R}^{m\times d}$.
The subsequent output projection absorbs the head-major-to-segment-column conversion by pre-permuting its plaintext weight columns, as detailed in \hyperref[app:he-kernels:proj]{Appendix~A.2}, so no explicit FHE remap is needed.

For each block $\ell$, we reuse two banks.
First, we complexify the value block and build a $\Psi$ bank over token offsets:
\[
\begin{aligned}
\langle \mathbf{u}^{(\ell)}\rangle
&=
\langle \mathbf{v}^{(\ell)}\rangle
- i\,\Psi^{m/2}\!\bigl(\langle \mathbf{v}^{(\ell)}\rangle\bigr),\\
\{\langle \mathbf{u}^{(\ell)}_t\rangle\}_{t=0}^{m/2-1}
&=
\mathsf{Bank}_{\Psi}\bigl(\langle \mathbf{u}^{(\ell)}\rangle,\{0,\dots,m/2-1\}\bigr).
\end{aligned}
\]
Second, let $\mathbf{n}_u$ be the binary segment mask that keeps exactly channel offset $u$ across the $H_{\mathrm{blk}}$ local heads in one block (\hyperref[app:he-kernels:qk]{Appendix~A.3}).
For each $t\in\{0,\dots,m/2-1\}$, we align and broadcast the folded-diagonal weights by
\[
\langle \mathbf{b}^{(\ell)}_t\rangle
=
\sum_{u=0}^{d_h-1}
\left(
\Phi^{t-u}\!\bigl(\langle \mathbf{p}^{(\ell)}_{\mathrm{fd}}\rangle\bigr)\odot \mathbf{n}_u
\right).
\]
The block output is then
\[
\langle \mathbf{o}^{(\ell)}_{\mathrm{hm}}\rangle
=
\sum_{t=0}^{m/2-1}
\bigl(\langle \mathbf{u}^{(\ell)}_t\rangle \odot \langle \mathbf{b}^{(\ell)}_t\rangle\bigr),
\qquad \ell=0,\dots,B_V-1.
\]
Each $\langle \mathbf{o}^{(\ell)}_{\mathrm{hm}}\rangle$ remains in head-major layout, and the output projection consumes this list directly via plaintext weight pre-permutation.

\noindent\textbf{Putting it together.}
The score kernel consumes the segment-column packed $Q$ and $K$ from the projection kernel (SCP Rule~1) and exports the scores into $K_{\min}(S)=\lceil Hm^2/(2n)\rceil$ ciphertexts in minimal folded-diagonal packing at the FHE--MPC boundary (Rule~3). The internal per-$t$ representation $\{\langle \mathbf{S}_t\rangle\}$ is repacked into a sequential minimal stream before conversion; see \hyperref[app:he-kernels:qk]{Appendix~A.3}.
The MPC softmax block receives this stream, converts it to additive shares via the CKKS-to-MPC protocol, evaluates the MBMax softmax surrogate securely, and converts the result back to CKKS in minimal folded-diagonal packing ($K_{\min}(P)=\lceil Hm^2/(2n)\rceil$ ciphertexts, Rule~2).
The value kernel then combines the folded-diagonal softmax list $\langle P_{\mathrm{fd}}\rangle$ with the head-major value list $\{\langle \mathbf{v}^{(\ell)}\rangle\}_{\ell}$ without FHE-side repacking, and outputs $\langle O_{\mathrm{hm}}\rangle$, which the output projection consumes directly through pre-permuted weight columns (Rule~1).
This keeps the MPC--FHE boundary at minimal ciphertext count while moving the broadcast into the FHE kernel where it composes cleanly with the value layout.
Figure~\ref{fig:ccmm-s-stage} illustrates both attention kernels on a toy instance.
The top row shows the score kernel producing folded-diagonal score ciphertexts $\{\langle \mathbf{S}_t\rangle\}$ that are exported to MPC for softmax, and the bottom row shows the value kernel consuming the folded-diagonal softmax ciphertext list $\langle P_{\mathrm{fd}}\rangle$ together with the head-major value ciphertext list $\{\langle \mathbf{v}^{(\ell)}\rangle\}_{\ell}$ to output the head-major ciphertext list $\langle O_{\mathrm{hm}}\rangle$ encoding $O=PV$.

\begin{figure*}
  \centering
  \includegraphics[width=1\textwidth]{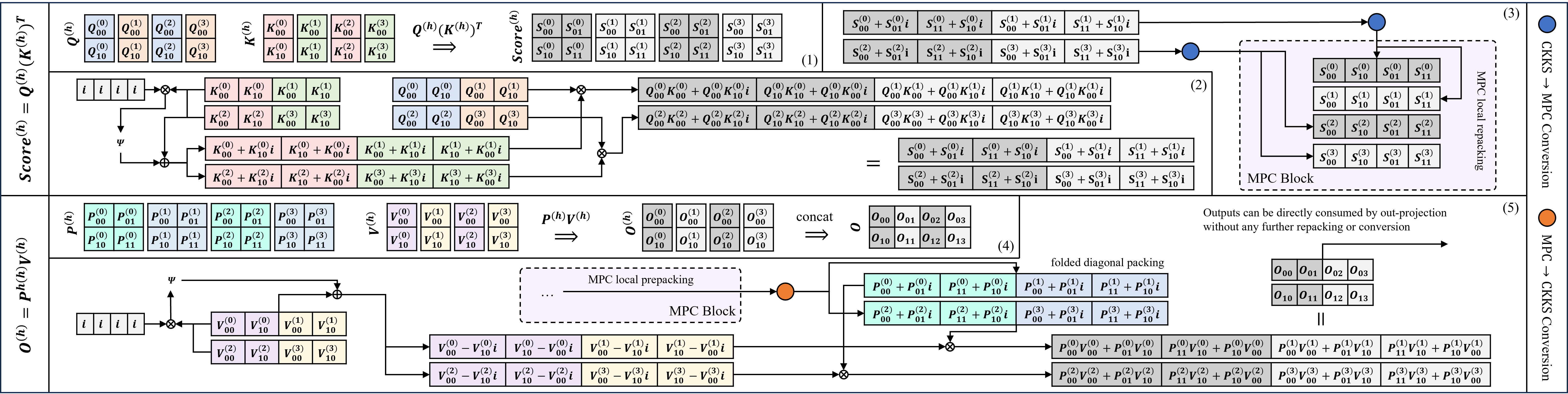}
  \caption{
    Toy example of the attention kernels with $H=4$, $m=2$, $d_h=1$, and $n=4$ slots.
    Each ciphertext contains $N_{\mathrm{seg}}=2$ segments.
    The top half shows the score kernel from segment-packed $Q$ and $K$ to folded-diagonal $S$, followed by MPC softmax that returns encrypted folded-diagonal $P$ in minimal packing.
    The bottom half shows the value kernel from folded-diagonal $P$ and head-major $V$ to the head-major output $O=PV$.
    (1) and (4) show the matrix view.
    (2), (3), and (5) illustrate the procedure at the slot level.
    }
  \label{fig:ccmm-s-stage}
\end{figure*}

Table~\ref{tab:matmulcc_bertbase} reports the key-switching counts of our attention kernels for BERT-base and compares them against prior work. The \system{} counts are measured from the implementation; the BOLT, Powerformer, and BLB counts are analytically derived from their published algorithms under the same tensor shapes ($m=128$, $d=768$, $H=12$, $n=16384$). Rather than summing rotations and ciphertext multiplications directly, the table also reports a weighted primitive-cost proxy in seconds using measured PhantomFHE A100 primitive latencies ($t_{\mathrm{rot}}=2.32$\,ms, $t_{\mathrm{ctmul}}=2.65$\,ms; Section~\ref{sec:backend_comparison}). \hyperref[app:he-kernels:qk]{Appendix~A.3} provides the asymptotic derivations.

\begin{table}[t]
\centering
\footnotesize
\setlength{\tabcolsep}{3pt}
\renewcommand{\arraystretch}{1.15}
\caption{Key-switching counts and weighted cost proxy for BERT-base attention ($n{=}16384$). Proxy: $\#\mathrm{rot}{\cdot}2.32{\times}10^{-3}{+}\#\mathrm{ctmul}{\cdot}2.65{\times}10^{-3}$\,s (PhantomFHE A100, Table~\ref{tab:he_primitive_bench}).}
\label{tab:matmulcc_bertbase}
\resizebox{\columnwidth}{!}{%
\begin{tabular}{@{}l|ccc|ccc@{}}
\toprule
& \multicolumn{3}{c|}{$Q^{(h)}(K^{(h)})^{\mathsf T}$}
& \multicolumn{3}{c}{$P^{(h)}V^{(h)}$} \\
& \#rot & \#mul & proxy\,(s) & \#rot & \#mul & proxy\,(s) \\
\midrule
BOLT~\cite{pang2024bolt}              & 13824 & 1536 & 36.1 & 21420 & 768  & 51.7 \\
Powerformer~\cite{park2025powerformer} & 4392  & 768  & 12.2 & 4882  & 1536 & 15.4 \\
BLB~\cite{xu2025blb}                  & 512   & 768  & 3.2  & 824   & 1536 & 6.0  \\
\system{}                              & 630   & 448  & 2.6  & 1524  & 384  & 4.6  \\
\bottomrule
\end{tabular}}
\end{table}

\section{Communication-Efficient MPC Blocks}
\label{sec:mpc}

\subsection{MPC Fixed-point Primitives}
\label{sec:mpc-primitives}

All MPC computation uses fixed-point arithmetic over $\mathbb{Z}_{2^\ell}$ with additive shares $[[\cdot]]_{2^\ell}$.
Real values are encoded with public fractional precision $F$ as $\lfloor r\cdot 2^F\rceil \in \mathbb{Z}_{2^\ell}$,
and all public constants and polynomial coefficients use the same scale.
Table~\ref{tab:mpc_ops} summarizes the MPC primitives used throughout.
We denote secure fixed-point multiplication by $\Pi_{\times}$, with truncation back to scale $2^F$.
\begin{table}[t]
\centering
\footnotesize
\setlength{\tabcolsep}{5pt}
\renewcommand{\arraystretch}{1.15}
\caption{MPC fixed-point primitives over $\mathbb{Z}_{2^\ell}$ with scale $2^F$.}
\label{tab:mpc_ops}
\begin{tabular}{@{}l|cc@{}}
\toprule
Primitive & Inputs & Output \\
\midrule
$\Pi_{\times}$ &
$\big([[a]]_{2^\ell},[[b]]_{2^\ell}\big)$ &
$\bigl[\!\bigl[\lfloor (a\cdot b)/2^F\rceil\bigr]\!\bigr]_{2^\ell}$ \\

$\Pi_{\mathrm{cmp}}$ &
$\big([[x]]_{2^\ell},\,\tau\big)$ &
$[[b]]_{2^\ell}$ where $b=\mathbf{1}[x<\tau]$ \\

$\Pi_{\mathrm{mux}}$ &
$\big([[u]]_{2^\ell},[[b]]_{2^\ell}\big)$ &
$[[v]]_{2^\ell}$ where $v=b\cdot u$ \\

$\Pi_{\mathrm{rowsum}}$ &
$[[\mathbf{X}]]_{2^\ell}\in\mathbb{Z}_{2^\ell}^{T\times d}$ &
$[[\mathbf{s}]]_{2^\ell}\in\mathbb{Z}_{2^\ell}^{T\times 1}$,
\;\; $s_t=\sum_{k=0}^{d-1} X_{t,k}$ \\

$\Pi_{\mathrm{bcast}}$ &
$\big([[\boldsymbol{\mu}]]_{2^\ell}\in\mathbb{Z}_{2^\ell}^{T\times 1},\,d\big)$ &
$[[\mathbf{M}]]_{2^\ell}\in\mathbb{Z}_{2^\ell}^{T\times d}$,
\;\; $M_{t,k}=\mu_t$ \\
\bottomrule
\end{tabular}
\end{table}

\subsection{Protocols for Non-linear Layers}

The nonlinear approximation functions for Softmax and LayerNorm follow Powerformer~\cite{park2025powerformer}, and the GELU piecewise approximation follows BOLT~\cite{pang2024bolt}. Our contribution is the MPC protocol realization of these approximations. By expressing the batch nonlinearities as short chains of MPC primitives, \system{} uses 7 interactive MPC rounds per layer: 3 for MBMax (Softmax), 4 for GELU, and 0 for MBLN, since it reduces to local operations on additive shares. \hyperref[app:mpc_round_comparison]{Appendix~E} compares these counts against prior systems.

\subsubsection{MPC MBMax Protocol}
We replace Softmax with the Batch Power-Max approximation from Powerformer using a power map and a fixed public normalizer.
A single public constant $R_d$ is distilled offline and reused during inference together with the public shift $c$ and exponent $p=5$.
The approximation is
\begin{equation}
\label{eq:mbmax_form}
\mathbf{Y} \;=\; \frac{(\mathbf{X}+c)^{p}}{R_d},
\qquad p=5,
\end{equation}
applied elementwise.
Algorithm~\ref{alg:secure-mbmax} realizes MBMax as a secure MPC protocol using a short multiplication chain and a final multiplication by the public fixed point constant $1/R_d$.

\begin{algorithm}[t]
\caption{Secure MBMax $\Pi_{\mathrm{MBMax}}([[\mathbf{X}]])$}
\label{alg:secure-mbmax}
\begin{algorithmic}[1]\raggedright
\Require $P_0,P_1$ hold $[[\mathbf{X}]]_{2^\ell}$ for $\mathbf{X}\in\mathbb{Z}_{2^\ell}^{m\times m}$ at scale $2^F$, public scalars $c$ and $R_d$
\Ensure $P_0,P_1$ obtain $[[\mathbf{Y}]]_{2^\ell}$ where $\mathbf{Y}=(\mathbf{X}+c)^{5}/R_d$ applied elementwise
\State $[[\bar{\mathbf{X}}]]_{2^\ell} \leftarrow [[\mathbf{X}]]_{2^\ell} + c$
\State $[[\bar{\mathbf{X}}^{2}]]_{2^\ell} \leftarrow \Pi_{\times}\!\left([[\bar{\mathbf{X}}]]_{2^\ell},[[\bar{\mathbf{X}}]]_{2^\ell}\right)$
\State $[[\bar{\mathbf{X}}^{4}]]_{2^\ell} \leftarrow \Pi_{\times}\!\left([[\bar{\mathbf{X}}^{2}]]_{2^\ell},[[\bar{\mathbf{X}}^{2}]]_{2^\ell}\right)$
\State $[[\bar{\mathbf{X}}^{5}]]_{2^\ell} \leftarrow \Pi_{\times}\!\left([[\bar{\mathbf{X}}^{4}]]_{2^\ell},[[\bar{\mathbf{X}}]]_{2^\ell}\right)$
\State $[[\mathbf{Y}]]_{2^\ell} \leftarrow (1/R_d)\cdot [[\bar{\mathbf{X}}^{5}]]_{2^\ell}$
\State \textbf{Output} $[[\mathbf{Y}]]_{2^\ell}$
\end{algorithmic}
\end{algorithm}

\subsubsection{MPC MBLN Protocol}
We adopt the Batch LayerNorm approximation from Powerformer, which replaces variance-dependent normalization with a fixed public scale.
For a row $x\in\mathbb{R}^d$,
\begin{equation}
\label{eq:mbln_form}
y \;=\; \gamma \odot \frac{x-\mu}{l\cdot R_d} + \beta,
\qquad
\mu=\mathrm{RowMean}(x).
\end{equation}
We fold the constant scale into $\tilde{\gamma}=\gamma/(l\cdot R_d)$ so that MBLN reduces to computing the row mean, centering the row, and applying a public elementwise affine map, which Algorithm~\ref{alg:secure-mbln} realizes as a secure MPC protocol.

\begin{algorithm}[t]
\caption{Secure MBLN $\Pi_{\mathrm{MBLN}}([[\mathbf{X}]])$}
\label{alg:secure-mbln}
\begin{algorithmic}[1]\raggedright
\Require $P_0,P_1$ hold $[[\mathbf{X}]]_{2^\ell}$ for $\mathbf{X}\in\mathbb{Z}_{2^\ell}^{T\times d}$ at scale $2^F$,
public $\tilde{\gamma},\beta\in\mathbb{Z}_{2^\ell}^{d}$ with $\tilde{\gamma}=\gamma/(l\cdot R_d)$
\Ensure $P_0,P_1$ obtain $[[\mathbf{Y}]]_{2^\ell}$ where each row satisfies $y=\tilde{\gamma}\odot(x-\mu)+\beta$
\State $[[\mathbf{s}]]_{2^\ell} \leftarrow \Pi_{\mathrm{rowsum}}([[\mathbf{X}]]) \in \mathbb{Z}_{2^\ell}^{T\times 1}$
\State $[[\boldsymbol{\mu}]]_{2^\ell} \leftarrow (1/d)\cdot [[\mathbf{s}]]_{2^\ell} \in \mathbb{Z}_{2^\ell}^{T\times 1}$
\State $[[\mathbf{Z}]]_{2^\ell} \leftarrow [[\mathbf{X}]]_{2^\ell} -
\Pi_{\mathrm{bcast}}([[\boldsymbol{\mu}]]_{2^\ell},d)\in\mathbb{Z}_{2^\ell}^{T\times d}$
\State $[[\mathbf{Y}]]_{2^\ell} \leftarrow \tilde{\gamma} \odot [[\mathbf{Z}]]_{2^\ell} + \beta$
\State \textbf{Output} $[[\mathbf{Y}]]_{2^\ell}$
\end{algorithmic}
\end{algorithm}

\subsubsection{GELU}
We evaluate GELU using the piecewise approximation and secure selection protocol from BOLT~\cite{pang2024bolt}, following the standard hybrid split used by BLB~\cite{xu2025blb}.
Concretely, CKKS forms the two candidate ciphertexts $F_0(\mathbf{x})$ and $F_1(\mathbf{x})$, and MPC selects among $\{0,F_0(\mathbf{x}),F_1(\mathbf{x}),\mathbf{x}\}$ using $\Pi_{\mathrm{cmp}}$ and $\Pi_{\mathrm{mux}}$.
This requires converting both $F_0(\mathbf{x})$ and $F_1(\mathbf{x})$, which increases the inbound conversion count for the GELU block; full protocol details are given in \hyperref[app:gelu_protocols]{Appendix~B}.

\section{Boundary Co-Design and Secure Conversion}
\label{sec:complex-conversion}

In \system{}, boundary design is part of the system architecture rather than a post-processing detail. We therefore pair a cost model for boundary placement with a secure complex CKKS--MPC conversion protocol that reduces payload while preserving compatibility with real-valued MPC blocks.

\subsection{Cost Model for Boundary Design}
\label{subsec:conv-overhead}

Conversion overhead grows with the number of CKKS--MPC conversion pairs and with the total ciphertext payload converted at the FHE--MPC boundaries.
Both factors depend on the boundary packing and on CKKS parameters that determine the ciphertext size. 

Fix a hybrid inference decomposition with $R$ MPC blocks indexed by $r=1,\dots, R$.
Block $r$ consumes a complex-valued boundary tensor $x^{\rightarrow}_r$ and produces a complex-valued boundary tensor $x^{\leftarrow}_r$.
A hybrid design chooses a packing for every $x^{\rightarrow}_r$ and $x^{\leftarrow}_r$ and specifies where conversions occur.

\subsubsection{Baseline Conversions and Decision Rule}
We define the minimal-conversion baseline design $B$ by two properties.
First, $B$ performs exactly one entry conversion CKKS$\to$MPC and one exit conversion MPC$\to$CKKS for each MPC block, and performs no other conversions.
Therefore, $B$ uses exactly $R$ conversion pairs, which is the minimum possible for a fixed set of $R$ MPC blocks.
Second, $B$ uses minimal packing, defined in Section~\ref{sec:ks-he:packing}, at every boundary tensor so that each $x^{\rightarrow}_r$ and $x^{\leftarrow}_r$ is represented using exactly $K_{\min}(x)$ ciphertexts.
In a standard Transformer encoder layer, the nonlinearities consist of the attention softmax, two layer normalizations, and the feed-forward GELU activation (Figure~\ref{fig:block-structure}).
Accordingly, we decompose each layer into $R=4$ MPC blocks (Softmax, LN1, GELU, LN2), so the baseline uses exactly four FHE$\leftrightarrow$MPC conversion pairs per layer.

However, deviating from $B$ can still be beneficial because expanded packings or additional conversions can reduce CKKS-side or MPC-side work by enabling layouts and kernels with fewer key-switched operations.
These choices increase conversion overhead by increasing either the number of conversion pairs or the number of ciphertexts converted. Prior works do not provide an analytic criterion for when a non-minimal conversion strategy is worthwhile.
We model this trade-off explicitly and use it to decide when moving away from $B$ is worthwhile.
Figure~\ref{fig:arch_compare} summarizes the design space.

\begin{figure}[t]
  \centering
  \includegraphics[width=\linewidth]{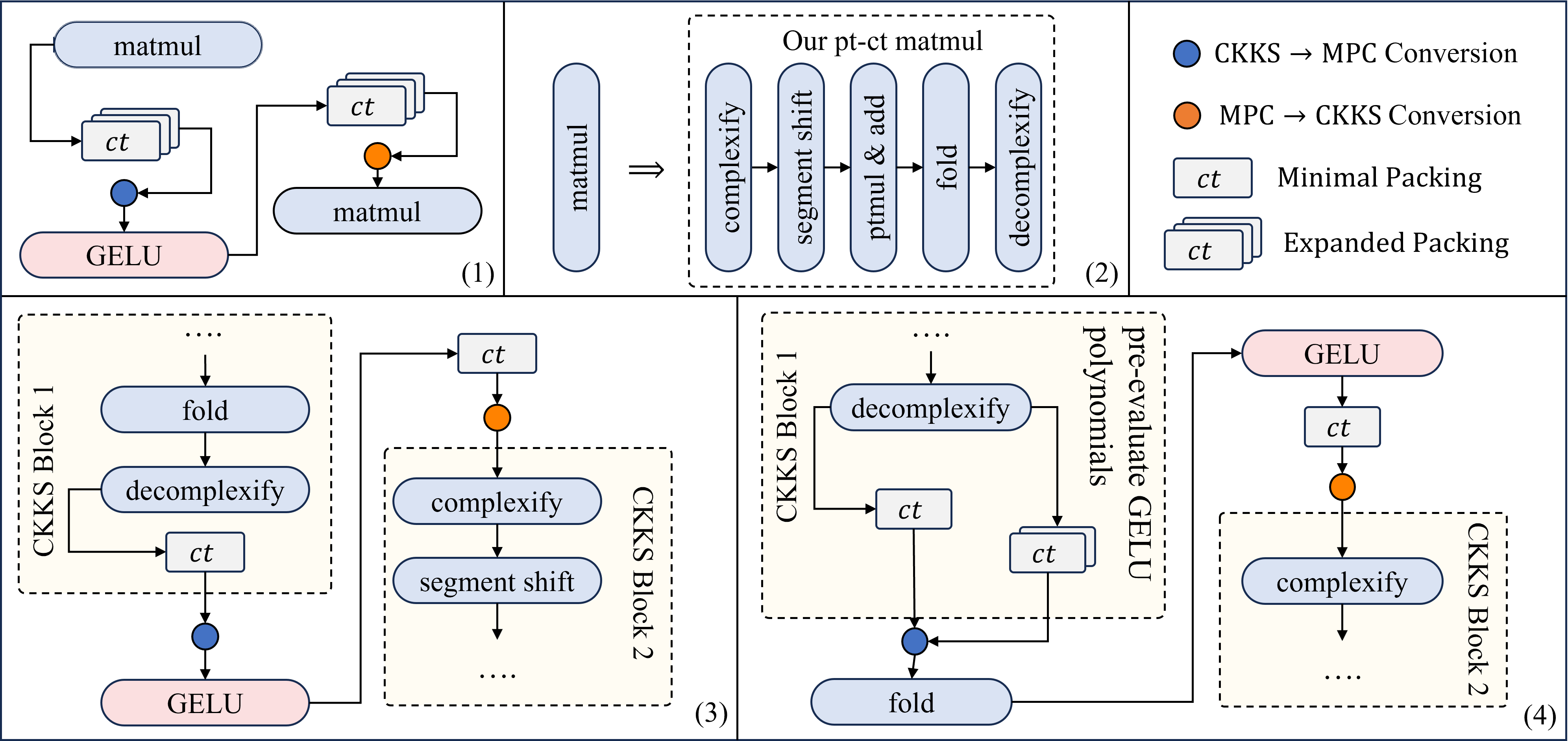}
\caption{
  Architecture comparison for hybrid inference.
  (1) Prior layerwise hybrid pipeline with expanded packing.
  (2) Component view of a plaintext--ciphertext projection kernel.
  (3) Minimal baseline using minimal packing at conversions.
  (4) \system{} pre-evaluates GELU polynomials in CKKS with expanded boundary packing for lower MPC cost.
  }
  \label{fig:arch_compare}
\end{figure}

We write end-to-end latency as
\begin{equation}
T_{\text{total}} \;=\; T_{\text{conv}} \;+\; T_{\text{ckks}} \;+\; T_{\text{mpc}}.
\label{eq:ttotal}
\end{equation}

We treat $T_{\text{conv}}$ as the aggregate boundary-conversion latency induced by the chosen conversion placement and boundary packing.
Thus $T_{\text{conv}}(D)$ captures both the number of conversion pairs in design $D$ and the total ciphertext payload converted at the boundaries.

Let $B$ denote the minimal-conversion baseline and $D$ a candidate design.
Define
\[
\begin{aligned}
\Delta T_{\text{conv}} &= T_{\text{conv}}(D)-T_{\text{conv}}(B),\\
\Delta T_{\text{comp}} &= \bigl(T_{\text{ckks}}(D)+T_{\text{mpc}}(D)\bigr)\\
                      &\quad{}-\bigl(T_{\text{ckks}}(B)+T_{\text{mpc}}(B)\bigr).
\end{aligned}
\]
We prefer $D$ if
\begin{equation}
\Delta T_{\text{conv}} \;+\; \Delta T_{\text{comp}} \;<\; 0 .
\label{eq:rule}
\end{equation}
This condition states that any increase in conversion overhead under $D$ is acceptable only when it is more than offset by a reduction in CKKS-side and MPC-side latency, resulting in a net decrease in end-to-end latency.

\subsubsection{Parameterization}
We parameterize $\Delta T_{\text{conv}}$ and $\Delta T_{\text{comp}}$ in terms of CKKS parameters, network profile, the total CKKS compute time difference $\Delta T_{\text{ckks}} = T_{\text{ckks}}(B)-T_{\text{ckks}}(D)$, and MPC round savings.
Let $K_{\text{extra}}$ denote the additional ciphertexts converted beyond the baseline, $R_{\text{extra}}$ the additional conversion round-trips, $\mathrm{ct\_bytes}(N, L)=2NL\cdot 8$ the serialized ciphertext size for ring degree $N$ with $L$ RNS limbs, and $\mathrm{BW}$, $\mathrm{RTT}$ the network bandwidth and round-trip time.
Then
\begin{equation}
\label{eq:dt_conv_param}
\Delta T_{\text{conv}} = \frac{K_{\text{extra}} \cdot \mathrm{ct\_bytes}(N,L)}{\mathrm{BW}} + R_{\text{extra}} \cdot \mathrm{RTT},
\end{equation}
\begin{equation}
\label{eq:dt_comp_param}
\Delta T_{\text{comp}} \approx -\bigl(
\Delta T_{\text{ckks}}
+ R_{\text{saved}} \cdot ({\mathrm{RTT}} + C_{\text{round}}/\mathrm{BW})
\bigr),
\end{equation}
where $\Delta T_{\text{ckks}}$ is the total CKKS compute time saved by $D$ relative to $B$ (positive when $D$ reduces FHE work), $R_{\text{saved}}$ is the number of MPC rounds saved, and $C_{\text{round}}$ is the MPC communication per round in bytes.
The per-round communication depends on the MPC backend: for Beaver-triple multiplication on $E$ elements with ring bit-width $\ell$, $C_{\text{round}} = E \cdot 4\lceil \ell/8 \rceil$.

Eqs.~\eqref{eq:dt_conv_param}--\eqref{eq:dt_comp_param} can be used to estimate the sign of $\Delta T_{\text{conv}}+\Delta T_{\text{comp}}$ and guide boundary placement. We treat them as a screening tool rather than a precise predictor, since $\Delta T_{\text{ckks}}$ and $C_{\text{round}}$ are measured under specific deployment conditions. Appendices~D and~H provide further discussion and validation details.

\subsection{Complex Conversion Protocol}
\label{sec:complex-conv}

Prior CKKS--MPC conversions in hybrid systems~\cite{boemer2020mp2ml,xu2025blb} use only the real part of each CKKS ciphertext, leaving the imaginary part unused.

At an FHE--MPC boundary, we convert between a CKKS ciphertext $\langle \mathbf{m}\rangle$ at scale $\Delta$ and MPC
additive shares over $\mathbb{Z}_{2^\ell}$.
We represent boundary values in a fixed point so that decryption and decoding yield a slot vector
$\mathbf{x}+i\mathbf{y}\in\mathbb{C}^{n}$ where $\mathbf{x},\mathbf{y}\in\mathbb{Z}^n$ fit in the signed range of
$\mathbb{Z}_{2^\ell}$.
Following BLB~\cite{xu2025blb}, we sample randomness directly in the polynomial ring modulo the coefficient modulus and then apply the CKKS encoding map, rather than sampling real numbers and encoding them. We write $A_{N,q}^2$ for the CKKS ciphertext space with ring degree $N$ and modulus $q$ and we set $q=Q_{\mathrm{conv}}$
at the boundary.

\subsubsection{Complex CKKS-to-MPC Conversion}
Let $Q_{\mathrm{conv}}$ be the ciphertext modulus at the boundary.
Given $\langle \mathbf{m}\rangle$ from $P_1$ that encrypts an encoding of $\mathbf{x}+i\mathbf{y}$, the conversion
outputs ring shares $[[\mathbf{x}]]_{2^\ell}$ and $[[\mathbf{y}]]_{2^\ell}$.
$P_1$ masks the boundary plaintext by sampling a uniform $\hat{\mathbf{r}}\leftarrow R_{Q_{\mathrm{conv}}}$ and sending
$\langle \mathbf{m}\rangle+\hat{\mathbf{r}}$ for decryption, so $P_0$ decrypts a uniformly masked plaintext-ring
element.
Protocol $\Pi_{\mathrm{C2M}}^{\mathbb{C}}$ in Algorithm~\ref{alg:complex-ckks-to-mpc} realizes this conversion.

\begin{algorithm}[t]
\caption{Complex CKKS-to-MPC Conversion $\Pi_{\mathrm{C2M}}^{\mathbb{C}}$}
\label{alg:complex-ckks-to-mpc}
\begin{algorithmic}[1]\raggedright
\Require $P_1$ holds $\langle \mathbf{m}\rangle \in A_{N,Q_{\mathrm{conv}}}^2$ at scale $\Delta$
\Ensure $\forall b\in\{0,1\}$, $P_b$ outputs $[[\mathbf{x}]]_{2^\ell}^b$ and $[[\mathbf{y}]]_{2^\ell}^b$
\State \textbf{$P_1$} sample $\hat{\mathbf{r}}\leftarrow R_{Q_{\mathrm{conv}}}$ uniformly, send
\Statex \hspace{\algorithmicindent}
$\langle \mathbf{d}\rangle \leftarrow \langle \mathbf{m}\rangle+\hat{\mathbf{r}}$ to $P_0$, and set
\Statex \hspace{\algorithmicindent}
$[[\hat{\mathbf{t}}]]^{Q_{\mathrm{conv}}}_1 \leftarrow -\hat{\mathbf{r}}$
\State \textbf{$P_0$} decrypt to get
\Statex \hspace{\algorithmicindent}
$[[\hat{\mathbf{t}}]]^{Q_{\mathrm{conv}}}_0 \leftarrow \mathrm{Dec}(\langle \mathbf{d}\rangle)$
\State \textbf{$P_0,P_1$} invoke
\Statex \hspace{\algorithmicindent}
$[[\hat{\mathbf{t}}]]^{2^\ell} \leftarrow
\Pi_{\mathrm{Field2Ring}}\!\left([[\hat{\mathbf{t}}]]^{Q_{\mathrm{conv}}}\right)$
\State \textbf{$\forall b\in\{0,1\}$} $P_b$ decode
\Statex \hspace{\algorithmicindent}
$\mathbf{z}_b \leftarrow \mathrm{Decode}([[\hat{\mathbf{t}}]]^{2^\ell}_b)$ and output
\Statex \hspace{\algorithmicindent}
$[[\mathbf{x}]]_{2^\ell}^b \leftarrow \Re(\mathbf{z}_b)$ and
$[[\mathbf{y}]]_{2^\ell}^b \leftarrow \Im(\mathbf{z}_b)$
\end{algorithmic}
\end{algorithm}

The two plaintext shares in Steps~1 and~2 reconstruct the unmasked boundary plaintext modulo $Q_{\mathrm{conv}}$ because
$(\hat{\mathbf{p}}+\hat{\mathbf{r}})+(-\hat{\mathbf{r}})\equiv \hat{\mathbf{p}}$ where
$\hat{\mathbf{p}}=\mathrm{Encode}(\mathbf{x}+i\mathbf{y},\Delta)$.
By correctness of $\Pi_{\mathrm{Field2Ring}}$ in \hyperref[app:ring-field]{Appendix~C}, Step~3 preserves the same signed fixed
point plaintext in $\mathbb{Z}_{2^\ell}$ except with negligible statistical error.
CKKS decoding is a linear map from the plaintext ring to the slot space, so decoding each additive plaintext share locally in
Step~4 yields additive shares of the decoded vector.
Therefore, the outputs reconstruct $\mathbf{x}+i\mathbf{y}$ up to the standard CKKS decoding error.
If the per-slot decoding error is below $1/2$ in both real and imaginary parts at the boundary scale, rounding recovers
the intended integers in $\mathbb{Z}_{2^\ell}$.
We empirically verify this condition on PhantomFHE ($n=2^{14}$, scale $\Delta=2^{40}$, 40-bit body primes): the maximum per-slot absolute error after one encrypt--decode round-trip is ${\sim}9\times10^{-6}$, well below $1/2$.

\subsubsection{Complex MPC-to-CKKS Conversion}
Given ring shares $[[\mathbf{x}]]_{2^\ell}$ and $[[\mathbf{y}]]_{2^\ell}$, the conversion outputs a CKKS ciphertext
$\langle \mathbf{m}\rangle$ at scale $\Delta$ that decrypts and decodes to a value close to $\mathbf{x}+i\mathbf{y}$.
We mirror prior real-only conversions by locally encoding the two share vectors as plaintext ring shares and then
switching these plaintext shares from $\mathbb{Z}_{2^\ell}$ to $\mathbb{Z}_{Q_{\mathrm{conv}}}$ using
$\Pi_{\mathrm{Ring2Field}}$.
Protocol $\Pi_{\mathrm{M2C}}^{\mathbb{C}}$ in Algorithm~\ref{alg:complex-mpc-to-ckks} realizes this conversion by
encrypting one plaintext share and adding the other plaintext share locally.

\begin{algorithm}[t]
\caption{Complex MPC-to-CKKS Conversion $\Pi_{\mathrm{M2C}}^{\mathbb{C}}$}
\label{alg:complex-mpc-to-ckks}
\begin{algorithmic}[1]\raggedright
\Require MPC shares $[[\mathbf{x}]]_{2^\ell},[[\mathbf{y}]]_{2^\ell}$ and scale $\Delta$ with boundary modulus $Q_{\mathrm{conv}}$
\Ensure $P_1$ outputs $\langle \mathbf{m}\rangle \in A_{N,Q_{\mathrm{conv}}}^2$ at scale $\Delta$
\State \textbf{$\forall b\in\{0,1\}$} $P_b$ center lift to
\Statex \hspace{\algorithmicindent}
$\tilde{\mathbf{x}}_b,\tilde{\mathbf{y}}_b$ and encode
\Statex \hspace{\algorithmicindent}
$[[\hat{\mathbf{t}}]]^{2^\ell}_b \leftarrow \mathrm{Encode}(\tilde{\mathbf{x}}_b+i\tilde{\mathbf{y}}_b,\Delta)$
\State \textbf{$P_0,P_1$} invoke
\Statex \hspace{\algorithmicindent}
$[[\hat{\mathbf{t}}]]^{Q_{\mathrm{conv}}} \leftarrow
\Pi_{\mathrm{Ring2Field}}\!\left([[\hat{\mathbf{t}}]]^{2^\ell}\right)$
\State \textbf{$P_0$} encrypt and send
\Statex \hspace{\algorithmicindent}
$\langle \mathbf{c}\rangle \leftarrow \mathrm{Enc}([[\hat{\mathbf{t}}]]^{Q_{\mathrm{conv}}}_0)$ to $P_1$
\State \textbf{$P_1$} output
\Statex \hspace{\algorithmicindent}
$\langle \mathbf{m}\rangle \leftarrow \langle \mathbf{c}\rangle + [[\hat{\mathbf{t}}]]^{Q_{\mathrm{conv}}}_1$
\end{algorithmic}
\end{algorithm}

By construction, the two encoded plaintext shares in Step~1 add to an encoding of $\mathbf{x}+i\mathbf{y}$ under the
signed fixed point interpretation in $\mathbb{Z}_{2^\ell}$.
By correctness of $\Pi_{\mathrm{Ring2Field}}$ in \hyperref[app:ring-field]{Appendix~C}, Step~2 converts these plaintext shares to
$R_{Q_{\mathrm{conv}}}$ while preserving their sum modulo $Q_{\mathrm{conv}}$ except with negligible statistical error.
After Step~4, the ciphertext encrypts the sum of the two plaintext shares, so decryption and decoding recover
$\mathbf{x}+i\mathbf{y}$ up to the standard CKKS decoding error.
If the per-slot decoding error is below $1/2$ in both parts at the boundary scale, a subsequent CKKS-to-MPC conversion
recovers the intended integers in $\mathbb{Z}_{2^\ell}$ under rounding.

\subsubsection{Security of Complex Conversion}
We consider semi-honest adversaries.
In $\Pi_{\mathrm{C2M}}^{\mathbb{C}}$, $P_0$ receives a masked ciphertext and decrypts a plaintext that is uniform in
$R_{Q_{\mathrm{conv}}}$ because $\hat{\mathbf{r}}$ is uniform and unknown to $P_0$, so the decryption result is
statistically independent of the encoded message.
In $\Pi_{\mathrm{M2C}}^{\mathbb{C}}$, $P_1$ receives only a CKKS ciphertext encrypting $P_0$'s plaintext share, which
reveals no information by semantic security, and then adds its own plaintext share locally.
By composition, if $\Pi_{\mathrm{Field2Ring}}$ and $\Pi_{\mathrm{Ring2Field}}$ are secure, then the conversion
protocols leak nothing beyond their prescribed outputs and public parameters.

\subsubsection{Communication Overhead and Compatibility}
A CKKS ciphertext has $n$ complex slots, corresponding to $2n$ real values across the real and imaginary parts.
A real-only conversion exports a single part and therefore transmits one ciphertext per real vector per direction.
Our complex conversion exports both parts so one ciphertext carries two independent real vectors, reducing the number of
ciphertexts transmitted per direction from $k$ to $\lceil k/2\rceil$ for $k$ real vectors.
The conversion outputs two real share vectors, so existing MPC protocols for real-valued nonlinearities apply by
processing the two parts independently.
For pipelines that otherwise treat CKKS ciphertexts as real, boundary-level complexify/decomplexify wrappers convert two real
ciphertexts into one complex ciphertext and back. Section~\ref{sec:he_microbench} shows that their overhead is small relative to the communication savings.

Table~\ref{tab:ct_io_swapped} compares analytically derived ciphertext counts per MPC block for the theoretical minimal baseline, prior hybrid systems, and \system{}. Prior papers do not report boundary payload in bytes, so we compare at the level of ciphertext counts under a common tensor shape and slot budget. Overall, \system{} reduces conversion ciphertext count by 25\% relative to BOLT and by 57\% relative to BLB, while still allowing an intentionally expanded GELU boundary when the PhantomFHE decision rule predicts a net latency win.

\begin{table}[t]
\centering
{\footnotesize
\setlength{\tabcolsep}{3pt}
\renewcommand{\arraystretch}{1.03}
\caption{Ciphertext conversion counts per MPC block ($n{=}16384$).}
\label{tab:ct_io_swapped}
\begin{tabular}{@{}l|cc|cc|cc|cc|cc|c@{}}
\toprule
\multicolumn{1}{c|}{} & \multicolumn{10}{c|}{MPC Blocks} & \multicolumn{1}{c}{} \\
\cmidrule(lr){2-11}
\multirow{2}{*}{Method}
& \multicolumn{2}{c|}{Softmax}
& \multicolumn{2}{c|}{\shortstack{Within\\Softmax}}
& \multicolumn{2}{c|}{LN1}
& \multicolumn{2}{c|}{GELU}
& \multicolumn{2}{c|}{LN2}
& \multirow{2}{*}{Total} \\
& in & out & in & out & in & out & in & out & in & out & -- \\
\midrule
Minimal  & 6  & 6  & -- & -- & 3 & 3 & 12 & 12 & 3 & 3 & 48 \\
BOLT     & 12 & 12 & -- & -- & 6 & 6 & 24 & 24 & 6 & 6 & 96 \\
BLB      & 12 & 12 & 12 & 12 & 6 & 6 & 72 & 24 & 6 & 6 & 168 \\
\system  & 6  & 6  & -- & -- & 3 & 3 & 36 & 12 & 3 & 3 & 72 \\
\bottomrule
\end{tabular}
}
\end{table}

\noindent\textbf{Relationship to Complex FHE Packing.}
BOLT and BLB use real-only conversion at FHE--MPC boundaries, transmitting one ciphertext per real vector~\cite{pang2024bolt,xu2025blb}. Powerformer uses complex packing inside its FHE kernels to halve key-switching counts via a generic blockwise CCMM~\cite{park2025powerformer,jiang2018secure}, but its output remains in column packing throughout.

\system{} differs in two ways. First, the attention kernels are not generic CCMM reductions but exploit the diagonal structure of $QK^\top$ and the token-shift structure of $PV$ to produce specialized folded-diagonal and head-major output formats that chain across stages via SCP without repacking. Second, \system{} uses complex packing at FHE--MPC boundaries, packing two real boundary tensors into one complex ciphertext and halving conversion payload. This boundary-level use is absent from Powerformer's FHE-only pipeline.

\noindent\textbf{Modulus Trimming Optimization.}
CKKS-to-MPC conversion payload depends on the serialized ciphertext size, which scales with the remaining modulus chain.
Ciphertexts should arrive at the boundary at the shortest admissible modulus level, since no further CKKS multiplications are needed.
If a ciphertext reaches the boundary at a higher level, $P_1$ can modulus switch it down before running
$\Pi_{\mathrm{C2M}}^{\mathbb{C}}$ to reduce transmitted payload and improve $T_{\mathrm{conv}}$ in
Eq.~\eqref{eq:ttotal}.
This optimization applies only to CKKS-to-MPC conversion since MPC-to-CKKS outputs must retain sufficient modulus for
subsequent CKKS evaluation.

Let $L_{\mathrm{conv}}$ be the target conversion level and set $q = Q_{L_{\mathrm{conv}}}$. We choose $L_{\mathrm{conv}}$ as the smallest level that keeps conversion statistically safe:
\[
\log_2(q) \ge \ell + \sigma + 1
\qquad\text{and}\qquad
q/2 > \Delta \cdot B_{\max}.
\]
$P_1$ repeatedly applies $\mathrm{ModSwitchToNext}$ (which removes the last RNS prime without rescaling) until reaching $L_{\mathrm{conv}}$, then runs $\Pi_{\mathrm{C2M}}^{\mathbb{C}}$ unchanged. This reduces only the transmission size; the MPC ring size and ring-to-field conversion cost are unaffected.

\section{Experimental Setup}
\label{sec:experiments}

We implement the main \system{} pipeline with PhantomFHE for CKKS and EzPC/SCI~\cite{chandran2019ezpc,rathee2020cryptflow2} for MPC. The supporting matched-backend comparison uses Liberate.FHE and CrypTen~\cite{knott2021crypten}. All FHE kernels run on NVIDIA A100 GPUs. MPC values use fixed-point encoding with $F=13$ fractional bits over $\mathbb{Z}_{2^\ell}$ with $\ell=43$. Table~\ref{tab:he_params} lists the CKKS parameters per FHE block, chosen so that the modulus chain is close to minimal at each conversion boundary.

\begin{table}[t]
\centering
\footnotesize
\setlength{\tabcolsep}{3.5pt}
\renewcommand{\arraystretch}{1.05}
\caption{CKKS parameterization per FHE block.}
\label{tab:he_params}
\begin{tabular}{@{}l|cccc@{}}
\toprule
 & Block 1 & Block 2 & Block 3 & Block 4 \\
\midrule
$N$   & 32768 & 32768 & 32768 & 65536 \\
Depth & 10    & 7     & 6     & 4     \\
Scale & $2^{42}$ & $2^{42}$ & $2^{40}$ & $2^{40}$ \\
\bottomrule
\end{tabular}
\end{table}

Following~\cite{pang2024bolt,lu2025bumblebee,xu2025blb}, we evaluate GPT2-base, BERT-base, and BERT-large. For accuracy, we use the GLUE tasks MRPC, RTE, and SST-2~\cite{wang2018glue}. We start from public Hugging Face checkpoints~\cite{huggingface}, fine-tune them following Powerformer-style training, and distill surrogate parameters for encrypted inference; training details are given in \hyperref[app:powerformer_training]{Appendix~J}.

Unless otherwise specified, we report four standard network profiles: LAN at 1\,Gbps with 0.3\,ms RTT, WAN1 at 400\,Mbps with 4\,ms RTT, WAN2 at 100\,Mbps with 4\,ms RTT, and WAN3 at 100\,Mbps with 80\,ms RTT.

\noindent\textbf{Latency Methodology.}
All FHE kernel and MPC local computation times are measured on a single A100 GPU, reported as the mean of three independent trials after one warm-up iteration. LAN and WAN latencies are computed by adding measured local compute time to a network overlay derived from measured byte counts and round counts under each profile, following the same methodology as prior work~\cite{pang2024bolt,xu2025blb,lu2025bumblebee}.

\noindent\textbf{Baseline Sources.}
The baseline metrics for prior systems are taken from the corresponding papers under their stated experimental settings~\cite{pang2024bolt,lu2025bumblebee,xu2025blb,moon2025thor,park2025powerformer}.
When a paper does not directly report a particular aggregate but provides a complete protocol specification, we derive the missing quantity by counting operations under the same tensor shapes. For example, the BOLT and BLB entries in Table~\ref{tab:ct_io_swapped} are obtained by analyzing their reported MPC protocols and boundary ciphertext dimensions.

\section{Evaluation}
\label{sec:performance}

\subsection{Microbenchmarks}
\label{sec:he_microbench}

\subsubsection{FHE Library Comparison}
To contextualize the runtimes reported in Section~\ref{sec:e2e_latency}, we benchmark the two GPU CKKS backends used in this paper, PhantomFHE~\cite{yang2024phantom} and Liberate.FHE~\cite{LiberateFHE}, under a common parameter set.
Table~\ref{tab:he_primitive_bench} reports the latency of core CKKS primitives, where PhantomFHE is on average about an order of magnitude faster than Liberate.FHE, consistent with prior library comparisons~\cite{domingo2025fideslib}.
PhantomFHE does not support bootstrapping, while Liberate.FHE does, which is why FHE-only baselines that require bootstrapping are evaluated under Liberate.FHE in Table~\ref{tab:breakdown_thor_powerformer_ours}.

\begin{table}[t]
\centering
\footnotesize
\setlength{\tabcolsep}{5.2pt}
\renewcommand{\arraystretch}{1.15}
\caption{CKKS primitive latency (ms) on A100 GPU ($N{=}65536$, depth 16).}
\label{tab:he_primitive_bench}
\begin{tabular}{l|rrrrr}
\toprule
Library & add & pt mul & ct mul & rot & conj \\
\midrule
PhantomFHE     & 0.07 & 0.36 & 2.65 & 2.32 & 2.33 \\
Liberate.FHE   & 3.04 & 13.72 & 48.66 & 32.70 & 22.73 \\
\bottomrule
\end{tabular}
\end{table}

\subsubsection{Multiplicative Depth Sensitivity}
Figure~\ref{fig:he_depth_latency} shows that CKKS primitive latency increases with multiplicative depth for both libraries, with higher and more variable latency for Liberate.FHE.
This supports our design choice to keep FHE depth low and refresh via FHE--MPC conversions.

\begin{figure}[t]
  \centering
  \includegraphics[width=1\linewidth]{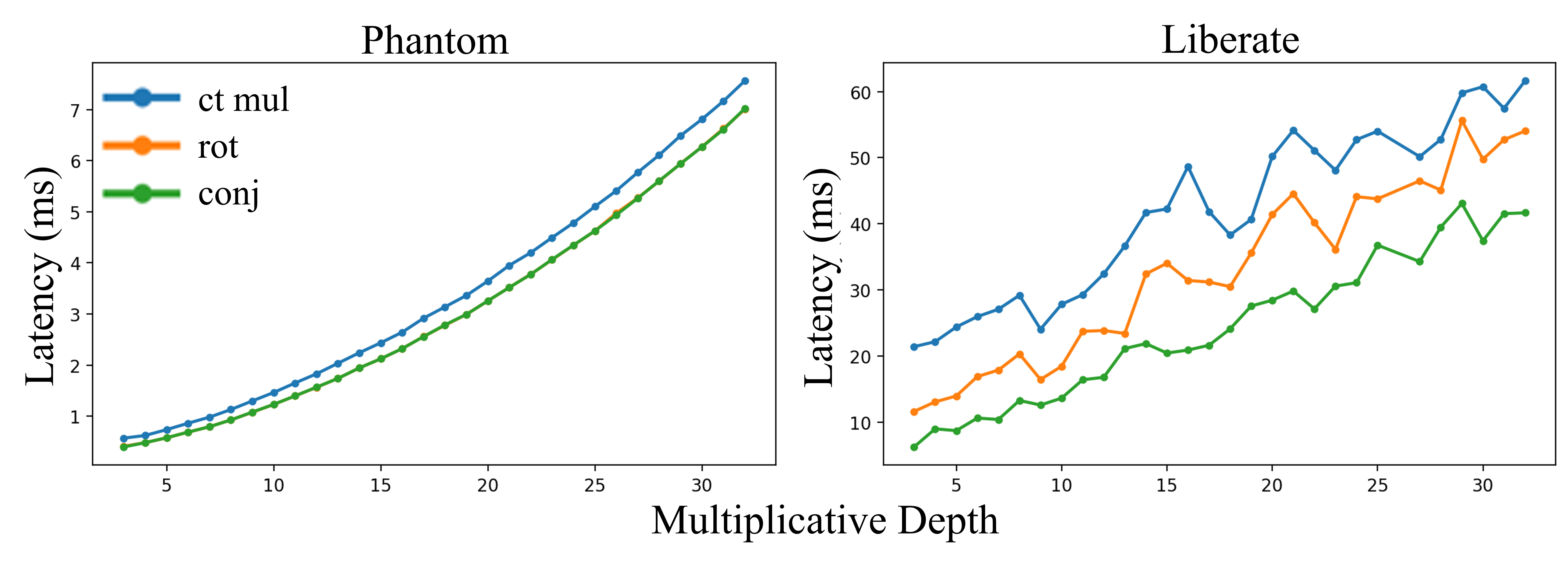}
  \caption{CKKS primitive latency vs.\ multiplicative depth for PhantomFHE and Liberate.FHE.}
  \label{fig:he_depth_latency}
\end{figure}

\subsubsection{Conversion Protocol Comparison}
We benchmark CKKS--MPC conversion for two real-valued ciphertexts under four standard network profiles.
The real baseline protocol~\cite{boemer2020mp2ml,pang2024bolt,xu2025blb} converts the ciphertexts independently and transmits two ciphertexts per direction.
Our complex variant packs the two vectors into the real and imaginary parts of one ciphertext and transmits one ciphertext per direction. We report end-to-end conversion latency, including CKKS-side complexification and decomplexification overhead, and network delay.
The total payload per conversion pair is 20.97\,MB for Real, 10.49\,MB for Complex, and 6.41\,MB for Complex plus trimming.
Table~\ref{tab:conversion_latency} reports latency and speedup across networks.

\begin{table}[t]
\centering
\footnotesize
\setlength{\tabcolsep}{4.2pt}
\renewcommand{\arraystretch}{1.16}
\caption{CKKS--MPC conversion latency under standard network profiles. $N{=}65536$, depth 8.}
\label{tab:conversion_latency}
\begin{tabular}{@{}l|cccc@{}}
\toprule
Network & LAN & WAN1 & WAN2 & WAN3 \\
\midrule
Real Time (ms)           & 168.97 & 435.43 & 1693.72 & 1997.72 \\
Complex Time (ms)        & 85.39  & 218.62 & 847.76  & 999.76 \\
Complex + Trim Time (ms) & 51.88  & 131.69 & 508.57  & 599.62 \\
\midrule
Speedup                  & 1.98$\times$ & 1.99$\times$ & 2.00$\times$ & 2.00$\times$ \\
+ Trim Speedup           & 3.26$\times$ & 3.31$\times$ & 3.33$\times$ & 3.33$\times$ \\
\bottomrule
\end{tabular}
\end{table}

The nearly $2\times$ reduction confirms that the complex protocol halves boundary communication, while trimming provides additional gain when the ciphertext reaches the boundary with excess modulus.

\subsection{End-to-end Evaluation}
\label{sec:e2e_latency}

\subsubsection{Latency and Communication}
Our main evaluation uses the PhantomFHE+EzPC pipeline and reports end-to-end latency together with online MPC communication. Section~\ref{sec:backend_comparison} compares the hybrid and FHE-only paradigms under a matched backend.

\begin{table}[t]
\centering
\scriptsize
\setlength{\tabcolsep}{3pt}
\renewcommand{\arraystretch}{1.12}
\caption{End-to-end latency (min). Prior-system numbers from their papers.}
\label{tab:e2e_latency}
\resizebox{\columnwidth}{!}{%
\begin{tabular}{@{}l|l|cccc|c@{}}
\toprule
Model & Framework & LAN & WAN1 & WAN2 & WAN3 & Comm (GB) \\
\midrule
\multirow{4}{*}{\makecell[l]{GPT2-base\\$m=64$}}
& BOLT       & 7.6 & 12.4 & 31.8 & 51.6 & 34.8 \\
& BumbleBee  & 2.9 &  3.5 &  4.9 & 12.0 &  2.5 \\
& BLB        & 2.0 &  2.5 &  3.9 &  8.1 &  1.5 \\
& \system    & \textbf{1.6} & \textbf{2.0} & \textbf{3.2} & \textbf{6.0} & \textbf{1.0} \\
\midrule
\multirow{4}{*}{\makecell[l]{BERT-base\\$m=128$}}
& BOLT       & 15.5 & 18.5 & 77.9 & 122.7 & 63.6 \\
& BumbleBee  &  4.3 &  6.2 & 11.6 & 21.3  & 5.8 \\
& BLB        &  2.5 &  3.8 & 6.6 & 13.2  & 3.0 \\
& \system    & \textbf{2.1} & \textbf{2.9} &  \textbf{4.6} & \textbf{8.9}  & \textbf{2.2} \\
\midrule
\multirow{4}{*}{\makecell[l]{BERT-large\\$m=128$}}
& BOLT       & 43.8 & 49.3 & 208.2 & 260.6 & 158.9 \\
& BumbleBee  &  9.7 & 15.6 &  30.8 &  49.2 &  15.2 \\
& BLB        &  6.6 &  9.8 &  16.2 &  24.9 &   7.8 \\
& \system    & \textbf{5.3} & \textbf{7.4} & \textbf{12.8} & \textbf{20.6} & \textbf{5.8} \\
\bottomrule
\end{tabular}
}
\end{table}

Following prior work~\cite{pang2024bolt,lu2025bumblebee,xu2025blb}, the Comm column reports only the online communication of the interactive MBMax and GELU protocols, excluding FHE--MPC conversion payload. We report boundary-conversion effects separately through Table~\ref{tab:ct_io_swapped}; the latency columns include all costs, including conversion. The total communication including conversion payload is approximately $1.1\times$ the online-only figure for \system{} under BERT-base; comparable multipliers apply to the baselines, as detailed in \hyperref[app:total_comm]{Appendix~H}. The relative ranking is preserved under total communication.

\system{} achieves the lowest latency and online MPC communication across all three models. Relative to BLB, \system{} is 1.3$\times$ faster on average while sending 1.4$\times$ less online MPC traffic. Relative to BumbleBee and BOLT, the corresponding average latency and communication gains are 2.1$\times$/2.6$\times$ and 9.8$\times$/30.4$\times$, respectively. Section~\ref{sec:backend_comparison} provides a matched-backend comparison that controls for library differences.

For context, plaintext inference on a single A100 GPU takes approximately 4.5\,ms for BERT-base, 9.1\,ms for BERT-large, and 4.2\,ms for GPT2-base. The secure pipeline is roughly $3{\times}10^4$ slower, which is consistent with the overhead reported by prior hybrid systems~\cite{pang2024bolt,lu2025bumblebee,xu2025blb}.

\subsubsection{Accuracy}
Table~\ref{tab:e2e_accuracy_powerformer} compares encrypted inference against a public uncased baseline (standard BERT with exact softmax, LayerNorm, and GELU) and a fine-tuned teacher that uses the same MBMax, MBLN, and approximate GELU surrogates as the encrypted pipeline. The uncased baseline is higher because it uses exact nonlinearities. The small gap of 0.25--0.35\% between the teacher and \system{} shows that the FHE+MPC execution itself introduces minimal quality loss, and that the remaining accuracy drop is due to the surrogate approximations, not the encrypted computation.

\begin{table}[t]
\centering
\scriptsize
\setlength{\tabcolsep}{5pt}
\renewcommand{\arraystretch}{1}
\caption{Accuracy (and F1 for MRPC) on selected GLUE tasks.}
\label{tab:e2e_accuracy_powerformer}
\resizebox{\columnwidth}{!}{%
\begin{tabular}{l|c|cc|ccc}
\toprule
Task & \#Val
& \shortstack{Uncased\\baseline}
& \shortstack{Fine-tuned\\teacher}
& \system{}
& \shortstack{$\Delta$ w.r.t.\\uncased}
& \shortstack{$\Delta$ w.r.t.\\teacher} \\
\midrule
SST-2 & 872 & 92.43\% & 92.13\% & 91.78\% & $-0.65\%$ & $-0.35\%$ \\
MRPC (acc)  & 408 & 87.75\% & 87.01\% & 86.76\% & $-0.99\%$ & $-0.25\%$ \\
MRPC (F1)   & 408 & 91.17\% & 90.82\% & 90.34\% & $-0.83\%$ & $-0.48\%$ \\
RTE   & 277 & 72.56\% & 69.43\% & 70.03\% & $-2.53\%$ & $+0.60\%$ \\
\bottomrule
\end{tabular}
}
\end{table}

\subsection{Hybrid vs.\ FHE-Only Paradigm Comparison}
\label{sec:backend_comparison}

A natural question is whether the hybrid FHE--MPC approach is worthwhile compared to FHE-only pipelines that avoid interaction entirely. To answer this, we implement \system{} with Liberate.FHE and CrypTen so that the comparison uses the same FHE backend and nonlinear approximations as the FHE-only baselines THOR and Powerformer.

Table~\ref{tab:breakdown_thor_powerformer_ours} shows the layer-wise breakdown. FHE-only systems offer a cleaner non-interactive solution with no communication during inference. However, if the deployment permits lightweight client--server interaction, the hybrid approach is substantially more efficient. Under LAN, eliminating bootstrapping yields a 3.5$\times$ speedup over THOR and 1.9$\times$ over Powerformer. The advantage narrows under WAN as MPC round-trip costs grow, so practitioners should choose based on their network conditions.

Even comparing only FHE linear components, \system{} is faster. For BERT-base, the FHE linear time per layer is 172.6\,s for \system{} versus 229.4\,s for THOR and 230.9\,s for Powerformer, thanks to SCP eliminating repacking and shallower depth reducing per-operation cost, as shown in Figure~\ref{fig:he_depth_latency}.

\begin{table}[t]
\centering
\scriptsize
\setlength{\tabcolsep}{4.4pt}
\renewcommand{\arraystretch}{1.12}
\caption{Layer-wise latency (s) on BERT-base (Liberate.FHE + CrypTen).}
\label{tab:breakdown_thor_powerformer_ours}
\resizebox{\columnwidth}{!}{%
\begin{tabular}{@{}l|cccccc@{}}
\toprule
Operation
& THOR
& \shortstack{Power\\former}
& \shortstack{Ours\\(LAN)}
& \shortstack{Ours\\(WAN1)}
& \shortstack{Ours\\(WAN2)}
& \shortstack{Ours\\(WAN3)} \\
\midrule
Attention layer      & 49.77  & 57.65 & 40.54 & 40.54 & 40.54 & 40.54 \\
Attention score      & 32.53  & 28.76 & 31.10 & 31.10 & 31.10 & 31.10 \\
Softmax              & 15.53  &  0.75 &  1.88 &  2.89 &  4.43 &  8.52 \\
Multi-head attention & 48.06  & 41.49 & 38.16 & 38.16 & 38.16 & 38.16 \\
LayerNorm1           &  7.13  &  0.37 &  0.98 &  1.23 &  1.50 &  3.45 \\
FC1                  & 49.80  & 59.21 & 32.50 & 32.50 & 32.50 & 32.50 \\
GELU                 & 29.42  &  8.31 &  3.38 &  6.24 & 14.64 & 29.98 \\
FC2                  & 49.19  & 43.77 & 30.27 & 30.27 & 30.27 & 30.27 \\
LayerNorm2           &  4.10  &  0.30 &  0.92 &  1.22 &  1.57 &  3.70 \\
Bootstrappings       & 337.86 & 103.72 & 0 & 0 & 0 & 0 \\
\midrule
Total                & 623.39 & 344.33 & 179.73 & 184.15 & 194.71 & 218.22 \\
\bottomrule
\end{tabular}
}
\end{table}

\subsection{Ablation Studies}
\label{sec:ablation}

\subsubsection{GELU Conversion Cost Analysis}
We study the inbound CKKS--MPC boundary of the GELU block and compare two pipeline variants that differ only in where the GELU polynomials are evaluated. In the minimal baseline, the pipeline converts only the GELU input tensor $\mathbf{x}$ into MPC using minimal packing, evaluates the polynomial candidates in MPC, and then applies $\mathrm{cmp}$ and $\mathrm{mux}$.
In \system's CKKS pre-evaluation variant, we evaluate the polynomial candidates, $F_0(\mathbf{x})$ and $F_1(\mathbf{x})$, in CKKS and convert $\mathbf{x}$ together with the candidates into MPC, so that MPC performs only $\mathrm{cmp}$ and $\mathrm{mux}$. \hyperref[app:gelu_protocols]{Appendix~B} describes the details of the two versions of GELU protocol adopted by \system{}.

Figure~\ref{fig:gelu_cost} reports the end-to-end latency difference between CKKS pre-evaluation and the baseline for BERT-base and BERT-large, decomposed into the added conversion overhead and the computation change under our cost model.
With PhantomFHE, the MPC cost reduction dominates on WAN1--WAN3, making CKKS pre-evaluation preferable under the decision rule in Eq.~\eqref{eq:rule}, while the net change on LAN is small, we use expanded packing for our main implementation pipeline.
With the alternative Liberate.FHE implementation, however, CKKS pre-evaluation incurs substantially higher CKKS-side cost, so the net improvement is no longer robust across network profiles.
Since the deployment network may be unknown a priori, we therefore adopt the no CKKS pre-evaluation GELU split for the results in Table~\ref{tab:breakdown_thor_powerformer_ours}.

\begin{figure}[t]
  \centering
  \includegraphics[width=\linewidth]{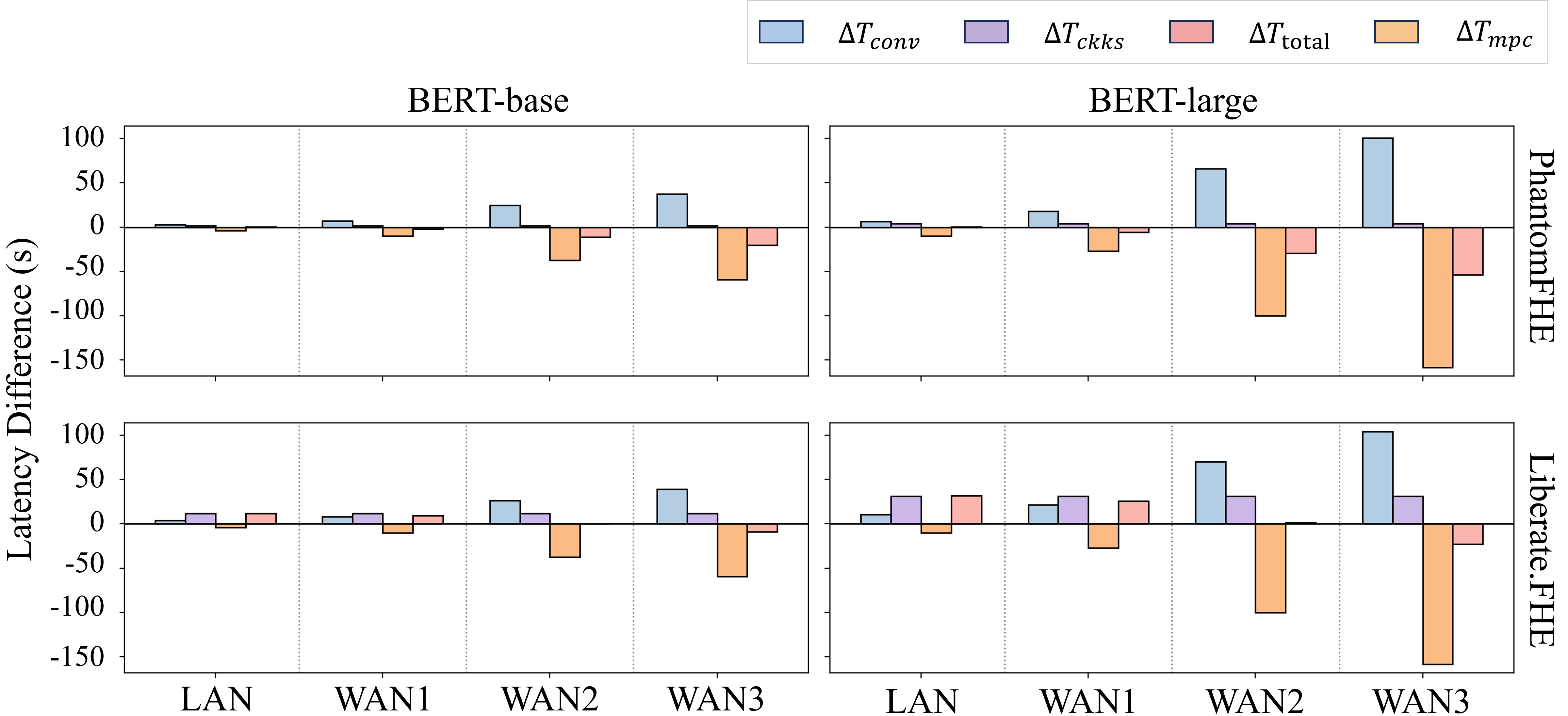}
  \caption{GELU latency deltas between \system{} and the minimal baseline, decomposed into computation savings and conversion overhead.}
  \label{fig:gelu_cost}
\end{figure}

\subsubsection{Contribution of Each Optimization}
Table~\ref{tab:latency_comm_blb_baseline_ablation} compares \system{} against BLB and two internal variants. In \texttt{w/o CC}, we disable complex conversion and remove GELU pre-evaluation. In \texttt{w/o SCP}, we keep the same kernels but force additional repacking between FHE stages. The comparison shows that \system{}'s gains come from lower online MPC traffic, fewer interactive rounds, and the elimination of cross-stage FHE remaps.

Without SCP, ciphertext repacking is needed at three direct FHE--FHE edges per layer: QKV$\to$Score, QKV$\to$Value, and Value$\to$OUT. Each edge remaps $B$ ciphertext blocks via the Halevi--Shoup slot-permutation decomposition~\cite{halevi2014algorithms}, requiring $\log_2 m$ rotations per block; see \hyperref[app:repack-algo]{Appendix~G.1} for details. For BERT-base with $m{=}128$ and $B{=}24$ total blocks across the three edges, this gives $24 \times 7 = 168$ extra rotations per layer, or $168 \times 2.32$\,ms $= 0.39$\,s per layer, totalling 4.7\,s over 12 layers. The observed \texttt{w/o SCP} gap is about 12\,s, which is larger because the realized remap also includes masking multiplications and recombination additions. The nearly constant slowdown across network profiles confirms that SCP primarily reduces FHE compute rather than communication.

\begin{table}[t]
\centering
\scriptsize
\setlength{\tabcolsep}{5pt}
\renewcommand{\arraystretch}{1.12}
\caption{Ablation of \system{} components against BLB.}
\label{tab:latency_comm_blb_baseline_ablation}
\resizebox{\columnwidth}{!}{%
\begin{tabular}{@{}l|l|cccc|c@{}}
\toprule
Model & Framework & LAN & WAN1 & WAN2 & WAN3 & Comm\ (GB) \\
\midrule
\multirow{4}{*}{\makecell[l]{GPT2-base\\$m=64$}}
& Baseline  & 2.0 & 2.5 & 3.9 & 8.1 & 1.5 \\
& w/o CC    & 1.8 & 2.3 & 3.6 & 7.4 & 1.3 \\
& w/o SCP   & 1.8 & 2.2 & 3.4 & 6.2 & 1.0 \\
& \system   & \textbf{1.6} & \textbf{2.0} & \textbf{3.2} & \textbf{6.0} & \textbf{1.0} \\
\midrule
\multirow{4}{*}{\makecell[l]{BERT-base\\$m=128$}}
& Baseline  & 2.5 & 3.8 & 6.6 & 13.2 & 3.0 \\
& w/o CC    & 2.3 & 3.5 & 6.0 & 11.4 & 2.7 \\
& w/o SCP   & 2.3 & 3.1 & 4.9 & 9.2 & 2.2 \\
& \system   & \textbf{2.1} & \textbf{2.9} & \textbf{4.6} & \textbf{8.9} & \textbf{2.2} \\
\midrule
\multirow{4}{*}{\makecell[l]{BERT-large\\$m=128$}}
& Baseline  & 6.6 & 9.8 & 16.2 & 24.9 & 7.8 \\
& w/o CC    & 6.3 & 9.2 & 15.4 & 23.6 & 7.3 \\
& w/o SCP   & 6.1 & 8.2 & 13.6 & 21.4 & 5.8 \\
& \system   & \textbf{5.3} & \textbf{7.4} & \textbf{12.8} & \textbf{20.6} & \textbf{5.8} \\
\bottomrule
\end{tabular}
}
\end{table}

\section{Related Work}
\label{sec:related}

We group prior secure machine learning inference systems into three categories: FHE-only, MPC-only, and hybrid FHE--MPC pipelines.

\subsection{FHE-only Secure Inference}
CryptoNets pioneered end-to-end leveled-FHE inference by replacing nonlinearities with low-degree polynomials~\cite{gilad-bachrach2016cryptonets}. CHET~\cite{dathathri2019chet} automates layout selection for CNN inference under FHE. For Transformers, THE-X systematizes polynomial approximations for softmax, layer normalization, and GELU to enable non-interactive BERT inference~\cite{chen2022thex}. CipherFormer~\cite{wang2024cipherformer} extends this to GPT-2 with CKKS packing optimizations. NEXUS~\cite{zhang2025nexus}, THOR~\cite{moon2025thor}, and Powerformer~\cite{park2025powerformer} further reduce overhead via CKKS-friendly packing, compression, and rotation-efficient attention schedules. A recent survey~\cite{li2025ptiSurvey} provides a comprehensive overview of private Transformer inference techniques.

FHE-only inference is non-interactive, but it is constrained by expensive ciphertext key-switching operations such as rotations and relinearization. It also often requires deep multiplicative depth with long modulus chains or heavy bootstrapping.

\subsection{MPC-only Secure Inference}
MPC-only systems evaluate the full model on secret shares, using arithmetic protocols for linear layers and specialized sub-protocols for nonlinear operations. SecureML~\cite{mohassel2017secureml}, MiniONN~\cite{liu2017minionn}, ABY~\cite{Demmler2015ABY}, and CryptFlow~\cite{kumar2020cryptflow} established practical CNN inference with fixed-point arithmetic and mixed-protocol execution. Falcon~\cite{wagh2021falcon} and XONN~\cite{riazi2019xonn} explore three-party and binary-circuit alternatives.

For Transformers and LLMs, MPCFormer~\cite{li2023mpcformer} and follow-on two-party systems co-design efficient protocols for attention, including softmax, GELU, and layer normalization, to scale to large hidden sizes and vocabularies~\cite{hou2023ciphergpt,dong2025puma,zeng2024securegpt,luo2024secformer,kei2025shaft,gupta2024sigma,huang2024secbert5}. MPC-only inference requires frequent interaction, so latency and bandwidth can become the bottleneck for deep models and wide activations.

\subsection{Hybrid FHE--MPC Secure Inference}
Hybrid systems typically execute wide linear algebra under FHE and delegate nonlinearities and selected reductions to MPC.
Gazelle~\cite{juvekar2018gazelle} introduced this split for CNNs, while Delphi~\cite{mishra2020delphi} and Cheetah~\cite{huang2022cheetah} improved latency through automated partitioning and optimized FHE and MPC building blocks.
For Transformers, Iron~\cite{hao2022iron} applies the split to BERT, BOLT~\cite{pang2024bolt} improves packing and conversion to avoid bootstrapping while preserving high-accuracy nonlinearities, BumbleBee~\cite{lu2025bumblebee} further optimizes hybrid execution with communication-efficient protocols and GPU-accelerated kernels, and BLB~\cite{xu2025blb} reduces boundary overhead via finer-grained operator fusion and fewer FHE--MPC conversions.
Hybrid pipelines balance ciphertext computation with interactive protocols, but they can inherit the overhead of both and are often dominated by additional FHE--MPC conversion costs at the boundaries.

\section{Conclusion}
We presented \system{}, a two-party hybrid FHE--MPC framework for private Transformer inference that frames packing, conversion, and nonlinear execution as a single co-design problem. \system{} uses Stage Compatible Patterns to keep adjacent FHE kernels layout-compatible, a secure complex CKKS--MPC conversion to reduce boundary payload, and a calibrated PhantomFHE decision rule to choose beneficial expanded boundaries.

This design reduces online communication and end-to-end latency across GPT2-base, BERT-base, and BERT-large relative to prior hybrid systems, with average gains of 30.4$\times$/9.8$\times$ against BOLT, 2.6$\times$/2.1$\times$ against BumbleBee, and 1.4$\times$/1.3$\times$ against BLB (communication/latency), while also outperforming FHE-only pipelines on BERT-base under a matched backend. On BERT-base, encrypted inference maintains near-plaintext accuracy on selected GLUE tasks. The broader lesson is that hybrid private inference should be designed around stage compatibility and boundary economics, not primitive choice alone.

\subsection{Limitations and Future Work}
Like prior systems that rely on surrogate nonlinearities~\cite{moon2025thor,park2025powerformer,xu2025blb}, \system{} still requires task-specific distillation. Extending the framework to settings with less task-specific retraining, stronger malicious-security guarantees, or richer backend-aware boundary models remains important future work.

\bibliographystyle{IEEEtran}
\bibliography{references}

\appendices
\clearpage
\section{Additional Details on FHE Kernels}
\label{app:he-kernels}

\subsection{Realization of Ciphertext Shifts}
\label{app:he-shifts}

We use the ciphertext shift idea of Jiang et al.~\cite{jiang2018secure} and instantiate it under our segment-stacked packing. Under this packing, logical shifts translate into CKKS rotations together with masks, so their cost depends on both rotation distance and the number of masked terms. Algorithms~\ref{alg:he-segshift} and~\ref{alg:he-intrashift} provide the FHE realizations of the inter-segment shift $\Phi^\Delta$ and the intra-segment shift $\Psi^t$, respectively.

\begin{algorithm}
  \caption{FHE realization of $\Phi^\Delta$}
  \label{alg:he-segshift}
  \begin{algorithmic}[1]\small
    \Require $\langle \mathbf{x}\rangle$ encrypts $\mathbf{x}\in\mathbb{C}^n$ with $n=m\cdot N_{\mathrm{seg}}$, shift $\Delta$
    \Ensure ciphertext encrypting $\Phi^\Delta(\mathbf{x})$
    \State \textbf{Output} $\textsf{rot}(\langle \mathbf{x}\rangle;\Delta m)$
  \end{algorithmic}
\end{algorithm}

\begin{algorithm}
  \caption{FHE realization of $\Psi^t$}
  \label{alg:he-intrashift}
  \begin{algorithmic}[1]\small
    \Require $\langle \mathbf{x}\rangle$ encrypts $\mathbf{x}\in\mathbb{C}^n$ with $n=m\cdot N_{\mathrm{seg}}$, shift $t$
    \Ensure ciphertext encrypting $\Psi^t(\mathbf{x})$
    \State $t \gets t \bmod m$
    \State $(s_1,s_2) \gets \bigl(t,(t-m)\bmod n\bigr)$
    \State Define $\mathbf{H},\mathbf{U}\in\{0,1\}^{m\times N_{\mathrm{seg}}}$ by
    \Statex \hspace{\algorithmicindent}
    $\mathbf{H}_{r,s}=\mathbf{1}\{0\le r<m-t\}\,;\quad
    \mathbf{U}_{r,s}=\mathbf{1}\{m-t\le r<m\}$
    \State $(\mathbf{h},\mathbf{u}) \gets \bigl(\operatorname{vec}(\mathbf{H}),\operatorname{vec}(\mathbf{U})\bigr)$
    \State \textbf{Output} $\bigl(\textsf{rot}(\langle \mathbf{x}\rangle;s_1)\odot \mathbf{h}\bigr) \;+\; \bigl(\textsf{rot}(\langle \mathbf{x}\rangle;s_2)\odot \mathbf{u}\bigr)$
  \end{algorithmic}
\end{algorithm}

Several kernels also require cyclic wrap-around within only the first $L$ slots, while keeping slots $L,\dots,n-1$ unchanged at zero. For a shift amount $\tau\in\{0,\dots,L-1\}$, we use two binary masks $\mathbf{a}_{L,\tau}$ and $\mathbf{b}_{L,\tau}$ to split the first $L$ slots into the portion that remains in place and the portion that wraps around. Concretely, we set
\[
  \begin{aligned}
    (\mathbf{a}_{L,\tau})_i &= \mathbf{1}\{0\le i < L-\tau\},\\
    (\mathbf{b}_{L,\tau})_i &= \mathbf{1}\{L-\tau\le i < L\}.
  \end{aligned}
\]
Algorithm~\ref{alg:he-rotfirst} realizes this restricted rotation using two CKKS rotations followed by masking and addition.

\begin{algorithm}
  \caption{FHE realization of $\mathsf{RotFirst}_{L}(\cdot;\tau)$}
  \label{alg:he-rotfirst}
  \begin{algorithmic}[1]\small
    \Require $\langle \mathbf{x}\rangle$ encrypts $\mathbf{x}\in\mathbb{C}^n$, integers $L\le n$, $\tau$
    \Ensure ciphertext encrypting the cyclic left shift of $\mathbf{x}$ by $\tau$ within slots $\{0,\dots,L-1\}$, leaving slots $L,\dots,n-1$ unchanged
    \State $\tau \gets \tau \bmod L$
    \State Define masks $\mathbf{a}_{L,\tau},\mathbf{b}_{L,\tau}\in\{0,1\}^n$ by
    \Statex \hspace{\algorithmicindent}
    $(\mathbf{a}_{L,\tau})_i=\mathbf{1}\{0\le i < L-\tau\}$,\quad
    $(\mathbf{b}_{L,\tau})_i=\mathbf{1}\{L-\tau\le i < L\}$
    \State $y_1 \gets \textsf{rot}(\langle \mathbf{x}\rangle;\tau)\odot \mathbf{a}_{L,\tau}$
    \State $y_2 \gets \textsf{rot}(\langle \mathbf{x}\rangle;(\tau-L)\bmod n)\odot \mathbf{b}_{L,\tau}$
    \State \textbf{Output} $y_1 + y_2$
  \end{algorithmic}
\end{algorithm}

When only the first $C=c_{\mathrm{used}}$ segments are active, meaning $L=Cm$ slots, we use the same restricted-rotation primitive to implement wrap-around modulo $C$ segments. Algorithm~\ref{alg:he-restrictsegshift} defines $\Phi_C^\Delta$ by applying $\mathsf{RotFirst}_{Cm}$ with shift $\Delta m$.

\begin{algorithm}
  \caption{FHE realization of $\Phi_C^\Delta$}
  \label{alg:he-restrictsegshift}
  \begin{algorithmic}[1]\small
    \Require $\langle \mathbf{x}\rangle$ encrypts $\mathbf{x}\in\mathbb{C}^n$ with $n=m\cdot N_{\mathrm{seg}}$, segment size $m$, active segments $C$, shift $\Delta$
    \Ensure ciphertext encrypting $\Phi_C^\Delta(\mathbf{x})$, i.e., wrap-around modulo $C$ active segments
    \State $(C,\Delta) \gets \bigl(\min(C, N_{\mathrm{seg}}),\ \Delta \bmod C\bigr)$
    \State $(L,\tau) \gets (Cm,\ \Delta m)$
    \State \textbf{Output} $\mathsf{RotFirst}_{L}\!\left(\langle \mathbf{x}\rangle;\tau\right)$
  \end{algorithmic}
\end{algorithm}

\subsection{Plaintext--Ciphertext Projection Kernel}
\label{app:he-kernels:proj}

We follow Section~3.2 for packing, dimensions, and the BSGS accumulation and fold.

\subsubsection{Plaintext Weights}
Under complexification, we use $U=\lceil G/2\rceil$ inputs
$\langle \tilde{\mathbf{x}}^{(u)}\rangle=\langle \mathbf{x}^{(2u)}+i\,\mathbf{x}^{(2u+1)}\rangle$ and we encode complex weights.
Let $\bar W$ denote the pre-permuted weight matrix used by this projection kernel. Its columns are ordered to match the target packing of the output.
The row index addresses input channels and the column index addresses output channels.
For output block $b$, giant index $p_g$, complex-pair index $u$, and baby index $q$, the value stored in active segment $c\in\{0,\dots,C-1\}$ is
\[
  \alpha(c,q)=(c+q)\bmod C,
  \qquad
  \beta(c,p_g)=(c-p_gN_1)\bmod C.
\]
\noindent Then
\[
  \begin{aligned}
    \tilde w^{(b)}_{u,p_g,q}(c)
    &=
    \bar W_{(2u)C + \alpha(c,q),\; bC + \beta(c,p_g)} \\
    &\quad - i\,\bar W_{(2u+1)C + \alpha(c,q),\; bC + \beta(c,p_g)}.
  \end{aligned}
\]
All segments $c\ge C$ store $0$.

Weights are stored as plaintext and can be cheaply pre-permuted to match downstream packing, which only changes plaintext constant construction and adds no FHE cost at inference time.
Under segment-column packing, segment $c$ of $\langle \mathbf{y}^{(b)}\rangle$ stores the length $m$ column
$Y_{:,\,bC+c}$.
A column permutation can therefore be realized by permuting plaintext weight columns so the projection emits $Y$ directly in the target order.

Let $d_2=Hd_h$ and index the attention columns by head $h\in\{0,\dots,H-1\}$ and within-head channel $u\in\{0,\dots,d_h-1\}$.
Write the canonical head-major column index as $k(h,u)=h d_h+u$.
For the score kernel, folded-diagonal packing assigns one segment per head, so it is convenient for $Q$ and $K$ to place the $H$ heads next to each other for each fixed $u$.
We use the score-friendly column order
\[
  \pi_{\mathrm{S}}(h,u)=uH+h.
\]
We pre-permute the plaintext weights by permuting columns as
\[
  \bigl(W_Q^{\pi_{\mathrm{S}}}\bigr)_{r,\,\pi_{\mathrm{S}}(h,u)}=(W_Q)_{r,\,k(h,u)},
\]
\[
  \bigl(W_K^{\pi_{\mathrm{S}}}\bigr)_{r,\,\pi_{\mathrm{S}}(h,u)}=(W_K)_{r,\,k(h,u)}.
\]
Then the projection emits $Q=AW_Q^{\pi_{\mathrm{S}}}$ and $K=AW_K^{\pi_{\mathrm{S}}}$ in the score-friendly segment order consumed by the score kernel.

For the value kernel, we use the value-friendly head-major order that keeps channels of the same head together. This aligns them with the head-stacked segments used in the folded-diagonal weights. We set $\pi_{\mathrm{V}}(h,u)=k(h,u)$ and use $W_V^{\pi_{\mathrm{V}}}=W_V$, so the projection emits $V=AW_V^{\pi_{\mathrm{V}}}$ in head-major order.
In all weight-encoding expressions above, the symbol $W$ refers to the chosen pre-permuted matrix, so no ciphertext repacking is needed between kernels.

\subsubsection{Key-Switching Complexity} Assuming a precomputed bank of $\Phi_C^\Delta$ with cost
$\#\textsf{rot}_{\mathrm{bank}}=O(UC)$ where $U=\lceil G/2\rceil$ as in Section~3.2,
each $\Phi_C^\Delta$ is implemented by $\mathsf{RotFirst}_{Cm}$ and costs two rotations.
The total key switching counts are
\[
  \#\textsf{rot}=O(B_{\mathrm{out}}N_2),
  \qquad
  \#\textsf{conj}=O(B_{\mathrm{out}}).
\]

\subsubsection{Special Fused $QK$ Projection} When two projections share the same ciphertext input, most notably $Q=AW_Q$ and $K=AW_K$, we fuse them into one plaintext--ciphertext projection under CKKS complex packing.
We construct complex plaintext weights $\tilde W$ by combining the two real weight matrices after applying the score-friendly column permutation.
Formally, we form $W_Q^{\pi_{\mathrm{S}}}$ and $W_K^{\pi_{\mathrm{S}}}$ and encode
\[
  \tilde W = W_Q^{\pi_{\mathrm{S}}} + i\,W_K^{\pi_{\mathrm{S}}}.
\]
Running one projection gives $\tilde Y=A\tilde W$ whose real part equals $Q$ and whose imaginary part equals $K$ in the score-friendly segment order consumed by the score kernel.
Fusion is implemented entirely in plaintext preprocessing, and the ciphertext-side procedure is identical to a single projection.
Figure~\ref{fig:cpmm1} shows a toy example of fused projection and weight pre-permutation.

\begin{figure}
  \centering
  \includegraphics[width=\linewidth]{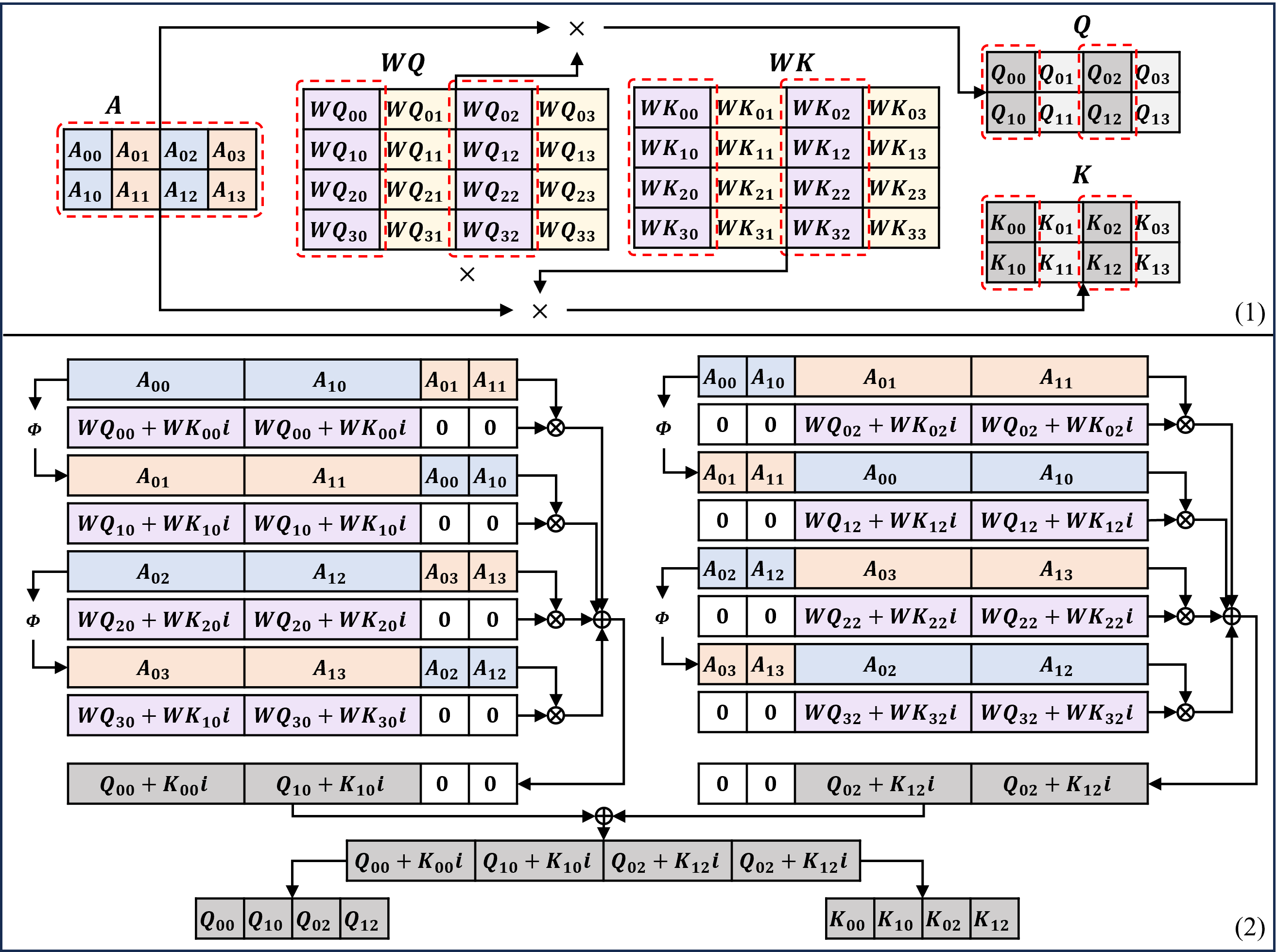}
  \caption{
    (1) Plain matrix multiplication for $Q=AW_Q$ and $K=AW_K$ with $A\in\mathbb{R}^{2\times 4}$ and $W_Q,W_K\in\mathbb{R}^{4\times 4}$.
    (2) Packed CKKS procedure with $n=4$ slots using fused plaintext weights pre-permuted $W_Q+iW_K$, producing $Q$ and $K$ in one projection.
  }
  \label{fig:cpmm1}
\end{figure}

\subsubsection{How the SCP Rules Are Satisfied}
For $V=AW_V$, the projection emits $V$ in the head-major packing expected by the value kernel.
Concretely, for each output block $b$, segment $c\in\{0,\dots,C-1\}$ of $\langle \mathbf{y}^{(b)}\rangle$ stores
$V_{:,\,bC+c}$, so the value kernel consumes the $V$ ciphertext list directly without any ciphertext-side repacking, satisfying Rule~1.
The value kernel consumes the MPC softmax output in minimal folded-diagonal packing, so its MPC-produced input attains
$K_{\min}(P)=\left\lceil \frac{H m^2}{2n}\right\rceil$, satisfying Rule~2.

The value kernel outputs the head-major ciphertext list $\langle O_{\mathrm{hm}}\rangle$. The output projection absorbs the head-major-to-segment-column conversion by pre-permuting its plaintext weight columns (no ciphertext repacking), and therefore consumes $\langle O_{\mathrm{hm}}\rangle$ directly, satisfying Rule~1.
The output projection is followed by an MPC block, so it exports its output in a minimal packing at the FHE--MPC boundary, satisfying Rule~3.

For the feed-forward network, let $\tilde{X}\in\mathbb{R}^{m\times d_2}$ denote the activation imported from MPC after LN1, and let $H\in\mathbb{R}^{m\times d_{\mathrm{ff}}}$ denote the activation output by MPC GELU.
Both FF1 and FF2 consume MPC-produced activations and we import them into CKKS in minimal packing, satisfying Rule~2, so the inputs attain
$K_{\min}(\tilde{X})=\left\lceil \frac{m d_2}{2n}\right\rceil$ and $K_{\min}(H)=\left\lceil \frac{m d_{\mathrm{ff}}}{2n}\right\rceil$.
FF1 computes the pre-activation $G=\tilde{X}W_1\in\mathbb{R}^{m\times d_{\mathrm{ff}}}$ and is followed by MPC GELU, so it exports $G$ in minimal packing, satisfying Rule~3, with
$K_{\min}(G)=\left\lceil \frac{m d_{\mathrm{ff}}}{2n}\right\rceil$.
FF2 computes $\tilde{X}'=HW_2\in\mathbb{R}^{m\times d_2}$ and is followed by an MPC block, so it exports $\tilde{X}'$ in minimal packing, satisfying Rule~3, with
$K_{\min}(\tilde{X}')=\left\lceil \frac{m d_2}{2n}\right\rceil$.

\subsection{Ciphertext--Ciphertext Attention Kernels}
\label{app:he-kernels:qk}

We follow Section~3.3 for the kernel structure and notation.

\subsubsection{Score Kernel Packing}
For each $t\in\{0,\dots,m/2-1\}$, the score kernel outputs a folded-diagonal ciphertext $\langle \mathbf{S}_t\rangle$.
Only the first $Hm$ slots are informative after head reduction, phase alignment, and undoing the baby shift.
Let $\mathrm{cut}(\langle \mathbf{x}\rangle)$ denote the first $Hm$ slots of $\langle \mathbf{x}\rangle$ and define the exported stream
$\mathbf{s}=\mathrm{cut}(\langle \mathbf{S}_0\rangle)\Vert\cdots\Vert \mathrm{cut}(\langle \mathbf{S}_{m/2-1}\rangle)$.
We pack $\mathbf{s}$ sequentially into the minimal number of ciphertexts
$K_{\min}(S)=\left\lceil \frac{Hm^2}{2n}\right\rceil$
(here $N(S)=Hm^2$ because the score matrix has $H$ heads each containing $m^2$ entries, and complex packing stores two real values per slot),
satisfying Rule~3 of SCP.

\subsubsection{Value Kernel Packing}
MPC returns the post-softmax weights in the same minimal folded-diagonal export order as above, namely as the stream
$\mathbf{p}=\mathrm{cut}(\langle \mathbf{P}_0\rangle)\Vert\cdots\Vert \mathrm{cut}(\langle \mathbf{P}_{m/2-1}\rangle)$
packed into $K_{\min}(P)=\left\lceil \frac{Hm^2}{2n}\right\rceil$ ciphertexts.
Let $H_{\mathrm{blk}}$ be the number of heads packed per value ciphertext and let $B_V=\lceil H/H_{\mathrm{blk}}\rceil$.
Under the standard setting $d_h=m/2=64$ used in our implementation, this minimal stream reshapes directly into a head-major folded-diagonal ciphertext list
\[
  \langle P_{\mathrm{fd}}\rangle=\{\langle \mathbf{p}^{(\ell)}_{\mathrm{fd}}\rangle\}_{\ell=0}^{B_V-1}.
\]
Within block $\ell$, segment
\[
  k(\tilde h,t)=\tilde h\,(m/2)+t
\]
with $\tilde h\in\{0,\dots,H_{\mathrm{blk}}-1\}$ and $t\in\{0,\dots,m/2-1\}$
stores the diagonal pair
\[
  p_{t}^{(\ell,\tilde h)} + i\,p_{t+m/2}^{(\ell,\tilde h)}.
\]
This is exactly the representation consumed by the value kernel, so the MPC--FHE boundary remains minimal and satisfies Rule~2 of SCP.

\subsubsection{Score Kernel Block Phase Correction}
Block $\ell$ has head phase $r_\ell=(\ell C)\bmod H$.
When $C\not\equiv 0 \pmod H$, different blocks start at different head offsets.
We sum score contributions by phase and align them inside the first $Hm$ slots.
For $r\in\{0,\dots,H-1\}$ define
\[
  \mathsf{Align}_r\!\left(\langle \mathbf{x}\rangle\right)
  =
  \mathsf{RotFirst}_{Hm}\!\left(\langle \mathbf{x}\rangle,(H-r)m\right).
\]
We apply $\mathsf{Align}_r$ before combining phases to form $\langle \mathbf{S}_t\rangle$.

\subsubsection{Mask Construction}
Let \(\mathbf{e}_s\) be the binary plaintext mask that is one on the \(m\) slots of segment \(s\) and zero elsewhere.
For the score kernel routing step, we use disjoint active-segment masks \(\{\mathbf{m}_c\}_{c=0}^{C-1}\).
When the first \(C\) segments are active, we set \(\mathbf{m}_c=\mathbf{e}_c\).

In the head-major layout, each ciphertext block packs \(H_{\mathrm{blk}}\) heads contiguously.
Segment \(k(\tilde h,u)=\tilde h\,d_h+u\) stores channel \(u\in\{0,\dots,d_h-1\}\) of local head \(\tilde h\in\{0,\dots,H_{\mathrm{blk}}-1\}\).
For the value-kernel broadcast step, define the channel mask
\[
  \mathbf{n}_u
  =
  \sum_{\tilde h=0}^{H_{\mathrm{blk}}-1}
  \mathbf{e}_{\tilde h\,d_h+u},
  \qquad
  u=0,\dots,d_h-1.
\]
This mask keeps exactly the segments with within-head channel offset \(u\) across all local heads in the block.
The value kernel output remains in head-major layout; the subsequent output projection absorbs the layout conversion via weight pre-permutation (no explicit FHE remap stage).

\subsubsection{Key Switching Complexity}
Let $B=\lceil d/C\rceil$ and $B_V=\lceil H/H_{\mathrm{blk}}\rceil$, where $C$ is the channel capacity per ciphertext and $H_{\mathrm{blk}}$ is the number of heads packed per ciphertext block.
Both kernels use precomputed $\Phi$ and $\Psi$ banks.

For the score kernel, each of the $m/2$ diagonals multiplies $B$ channel blocks.
Rotations are dominated by $\Psi$-bank construction and the per-diagonal postprocessing.
\[
  \begin{aligned}
    \#\textsf{ctmul}_S=\#\textsf{relin}_S &= O(mB),\\
    \#\textsf{rot}_S &= O\bigl(B(\beta+g)+m\log C\bigr).
  \end{aligned}
\]
For BERT-base with $B=7$, $\beta=16$, $g=8$, $m=128$, $C=120$, and caching of repeated shifts, the implementation measures 630 rotations and 448 ct-ct multiplications (Table~2).

For the value kernel, each of the $m/2$ offsets uses $B_V$ head blocks.
Rotations are dominated by the token-shift $\Psi$ bank over $m$ offsets and the $\Phi$ broadcasts over $O(d_h)$ segment shifts per block.
\[
  \begin{aligned}
    \#\textsf{ctmul}_V=\#\textsf{relin}_V &= O(mB_V),\\
    \#\textsf{rot}_V &= O\bigl(B_V(m+d_h)\bigr).
  \end{aligned}
\]
For BERT-base with $B_V=6$, $m=128$, $d_h=64$, the implementation measures 1524 rotations and 384 ct-ct multiplications.

\subsubsection{Why the KS Counts Differ from Powerformer}
Powerformer~\cite{park2025powerformer} accelerates attention by applying complex packing to a generic blockwise CCMM under modified column packing. Its blockwise transforms $\tilde\sigma,\tilde\tau,\tilde\phi,\tilde\psi$ and auxiliary routines such as \textsc{BlockTrans}, \textsc{BlockTau}, \textsc{BlockSigma}, \textsc{Extract}, and \textsc{SplitPaste} remain part of the algorithmic skeleton; complex arithmetic mainly reduces the number of calls to these routines by evaluating paired operands together. In that sense, the gain is a constant-factor reduction within the generic CCMM framework rather than a change to the object being computed.

By contrast, \system{} rewrites attention around the structure of the desired outputs. For the score matrix, we do not first form a generic packed product and then extract the needed diagonals. Instead, for $t=j\beta+s$ we directly produce the folded pair $(t,t+m/2)$ by
\[
\langle \mathbf{u}_t\rangle
=
\sum_{\ell}
\langle \mathbf{q}^{(\ell)}_{-s}\rangle
\odot
\bigl(
\langle \mathbf{k}^{(\ell)}_{j\beta}\rangle
+ i\,\langle \mathbf{k}^{(\ell)}_{m/2+j\beta}\rangle
\bigr).
\]
One complex $\textsf{ctmul}$ therefore produces two diagonals at once, and the required shifts come from reusable $\Psi$-banks of sizes $\beta$ and $g=m/\beta$ rather than from a full blockwise permutation-and-extraction pipeline. This is why the score kernel scales as $\#\textsf{ctmul}_S=O(mB)$ and $\#\textsf{rot}_S=O(B(\beta+g)+m\log C)$.

For the value kernel, we again avoid a generic CCMM by using the token-shift identity
\[
\langle \mathbf{u}^{(\ell)}\rangle
=
\langle \mathbf{v}^{(\ell)}\rangle
- i\,\Psi^{m/2}\!\bigl(\langle \mathbf{v}^{(\ell)}\rangle\bigr),
\]
so that the offset pair $(t,t+m/2)$ is handled by one complex operand
\(\langle \mathbf{u}^{(\ell)}_t\rangle=\Psi^{-t}(\langle \mathbf{u}^{(\ell)}\rangle)\).
This halves the token-shift bank from $m$ offsets to $m/2$ effective complex pairs, and the folded-diagonal weights are broadcast with simple head-major masks instead of generic block transforms. The resulting complexity is $\#\textsf{ctmul}_V=O(mB_V)$ and $\#\textsf{rot}_V=O(B_V(m+d_h))$.

The key point is therefore structural: Powerformer obtains a constant-factor compression of a generic matrix-multiplication pipeline, whereas \system{} rewrites attention into diagonal-pair and token-shift-pair primitives that match the geometry of $QK^\top$ and $PV$. Together with SCP eliminating post-kernel remaps, this is the main reason Table~2 shows substantially lower weighted primitive cost for \system{} than for Powerformer.

\section{Protocols for GELU}
\label{app:gelu_protocols}

We employ the GELU approximation and secure evaluation protocol from BOLT~\cite{pang2024bolt}. The approximation is a three-branch piecewise function that saturates outside a bounded region and uses a low-degree polynomial inside. Intuitively, values far in the positive tail are passed through unchanged, values far in the negative tail are clamped to zero, and only the mid-range requires a nontrivial computation. This structure is well-suited to MPC because it reduces the expensive arithmetic to the interval where GELU is most nonlinear, while the outer branches are implemented by inexpensive conditional selection.

For $x\in\mathbb{R}$,
\begin{equation}
  \label{eq:bolt_approx_gelu_piecewise}
  \mathrm{ApproxGELU}(x)=
  \begin{cases}
    x, & x>2.7,\\
    \begin{aligned}
      a|x|^4 + b|x|^3 + c|x|^2 \\
      {}+ d|x| + e + 0.5x,
    \end{aligned}
    & |x|\le 2.7,\\
    0, & x<-2.7.
  \end{cases}
\end{equation}

We follow the same MPC notation and primitives as in Section~4.1. In particular, values are represented as additive shares over $\mathbb{Z}_{2^\ell}$ and XOR-shared bits, and the protocol uses the standard comparison and multiplexing primitives to implement slotwise branching. Let $\mathbf{x}^{(0)},\mathbf{x}^{(1)}\in\mathbb{R}^n$ and pack $\mathbf{x}^{\mathbb{C}}=\mathbf{x}^{(0)}+i\mathbf{x}^{(1)}$.

Our complex conversion in Section~5 exports the real and imaginary channels in one invocation of $\Pi_{\mathrm{C2M}}^{\mathbb{C}}$ and packs the two MPC outputs back with $\Pi_{\mathrm{M2C}}^{\mathbb{C}}$.

Within MPC, the key idea is to evaluate the polynomial branch using ordinary arithmetic on shared values, while using comparisons and multiplexing to select the correct branch for each slot. The absolute value in the mid-range polynomial is handled by first determining the sign of each slot and then choosing between two sign-specialized polynomials that are identical up to the sign of the odd-degree terms. This avoids an explicit absolute-value circuit and reduces the branching logic to a small number of comparisons and multiplexers.

For $j\in\{0,1\}$, define candidate vectors
\begin{equation}
  \label{eq:bolt_gelu_candidates_vec}
  \begin{aligned}
    \mathbf{F}_0^{(j)} &= a(\mathbf{x}^{(j)})^4 - b(\mathbf{x}^{(j)})^3 + c(\mathbf{x}^{(j)})^2 + (0.5-d)\mathbf{x}^{(j)} + e,\\
    \mathbf{F}_1^{(j)} &= a(\mathbf{x}^{(j)})^4 + b(\mathbf{x}^{(j)})^3 + c(\mathbf{x}^{(j)})^2 + (0.5+d)\mathbf{x}^{(j)} + e.
  \end{aligned}
\end{equation}
We use the same polynomial coefficients $(a,b,c,d,e)$ as reported in BOLT~\cite{pang2024bolt}.

\begin{algorithm}[t]
  \caption{Secure GELU with CKKS pre-evaluation}
  \label{alg:secure-gelu-bolt-complex}
  \begin{algorithmic}[1]\raggedright
    \Require $P_1$ holds a complex CKKS ciphertext $\langle \mathbf{x}^{\mathbb{C}}\rangle$ at scale $\Delta$ encrypting
    $\mathbf{x}^{\mathbb{C}}=\mathbf{x}^{(0)}+i\mathbf{x}^{(1)}$ with $\mathbf{x}^{(0)},\mathbf{x}^{(1)}\in\mathbb{R}^n$
    \Require Public $a,b,c,d,e$ and thresholds $-2.7,0,2.7$ in MPC fixed-point
    \Ensure $P_1$ outputs $\langle \mathbf{y}^{\mathbb{C}}\rangle$ with $\mathbf{y}^{\mathbb{C}}\approx \mathbf{y}^{(0)}+i\mathbf{y}^{(1)}$,
    $\mathbf{y}^{(j)}\approx \mathrm{ApproxGELU}(\mathbf{x}^{(j)})$

    \State \textbf{CKKS:} extract real/imag ciphertexts
    \Statex \hspace{\algorithmicindent}
    $\langle \mathbf{x}^{(0)}\rangle \leftarrow \tfrac{1}{2}\bigl(\langle \mathbf{x}^{\mathbb{C}}\rangle+\textsf{conj}(\langle \mathbf{x}^{\mathbb{C}}\rangle)\bigr)$
    \Statex \hspace{\algorithmicindent}
    $\langle \mathbf{x}^{(1)}\rangle \leftarrow \tfrac{1}{2i}\bigl(\langle \mathbf{x}^{\mathbb{C}}\rangle-\textsf{conj}(\langle \mathbf{x}^{\mathbb{C}}\rangle)\bigr)$
    \State \textbf{CKKS:} for $j\in\{0,1\}$, using $\textsf{add}$, $\textsf{ptmul}$, and $\textsf{ctmul}$, compute
    \Statex \hspace{\algorithmicindent}
    $\langle (\mathbf{x}^{(j)})^2\rangle,\langle (\mathbf{x}^{(j)})^3\rangle,\langle (\mathbf{x}^{(j)})^4\rangle$ and form
    \Statex \hspace{\algorithmicindent}
    $\langle \mathbf{F}_0^{(j)}\rangle,\langle \mathbf{F}_1^{(j)}\rangle$ by Eq.~\eqref{eq:bolt_gelu_candidates_vec}

    \State \textbf{CKKS:} let $\mathbf{i}\in\mathbb{C}^n$ be the all-$i$ plaintext vector; pack
    \Statex \hspace{\algorithmicindent}
    $\langle \mathbf{F}_0^{\mathbb{C}}\rangle \leftarrow \langle \mathbf{F}_0^{(0)}\rangle + (\langle \mathbf{F}_0^{(1)}\rangle \odot \mathbf{i})$
    \Statex \hspace{\algorithmicindent}
    $\langle \mathbf{F}_1^{\mathbb{C}}\rangle \leftarrow \langle \mathbf{F}_1^{(0)}\rangle + (\langle \mathbf{F}_1^{(1)}\rangle \odot \mathbf{i})$

    \State \textbf{Convert:}
    $([[\mathbf{x}^{(0)}]]_{2^\ell},[[\mathbf{x}^{(1)}]]_{2^\ell})
    \leftarrow \Pi_{\mathrm{C2M}}^{\mathbb{C}}(\langle \mathbf{x}^{\mathbb{C}}\rangle)$
    \State \textbf{Convert:}
    $([[\mathbf{f}_0^{(0)}]]_{2^\ell},[[\mathbf{f}_0^{(1)}]]_{2^\ell})
    \leftarrow \Pi_{\mathrm{C2M}}^{\mathbb{C}}(\langle \mathbf{F}_0^{\mathbb{C}}\rangle)$
    \State \textbf{Convert:}
    $([[\mathbf{f}_1^{(0)}]]_{2^\ell},[[\mathbf{f}_1^{(1)}]]_{2^\ell})
    \leftarrow \Pi_{\mathrm{C2M}}^{\mathbb{C}}(\langle \mathbf{F}_1^{\mathbb{C}}\rangle)$

    \State \textbf{MPC:} for $j\in\{0,1\}$ set
    $[[\mathbf{b}_{-}^{(j)}]]\!\gets\!\Pi_{\mathrm{cmp}}([[\mathbf{x}^{(j)}]]_{2^\ell},-2.7)$,
    $[[\mathbf{b}_{0}^{(j)}]]\!\gets\!\Pi_{\mathrm{cmp}}([[\mathbf{x}^{(j)}]]_{2^\ell},0)$,
    $[[\mathbf{b}_{+}^{(j)}]]\!\gets\!\Pi_{\mathrm{cmp}}([[\mathbf{x}^{(j)}]]_{2^\ell},2.7)$
    \State \textbf{MPC:} for $j\in\{0,1\}$ derive indicators
    $[[\mathbf{z}_0^{(j)}]] \gets [[\mathbf{b}_{-}^{(j)}]]\oplus[[\mathbf{b}_{0}^{(j)}]]$,
    $[[\mathbf{z}_1^{(j)}]] \gets [[\mathbf{b}_{0}^{(j)}]]\oplus[[\mathbf{b}_{+}^{(j)}]]$,
    $[[\mathbf{z}_2^{(j)}]] \gets 1 \oplus [[\mathbf{b}_{+}^{(j)}]]$
    \State \textbf{MPC:} for $j\in\{0,1\}$ output
    $[[\mathbf{y}^{(j)}]] \gets
    \Pi_{\mathrm{mux}}([[\mathbf{f}_0^{(j)}]],[[\mathbf{z}_0^{(j)}]])
    +\Pi_{\mathrm{mux}}([[\mathbf{f}_1^{(j)}]],[[\mathbf{z}_1^{(j)}]])
    +\Pi_{\mathrm{mux}}([[\mathbf{x}^{(j)}]],[[\mathbf{z}_2^{(j)}]])$

    \State \textbf{Convert:} $\langle \mathbf{y}^{\mathbb{C}}\rangle \leftarrow \Pi_{\mathrm{M2C}}^{\mathbb{C}}([[\mathbf{y}^{(0)}]]_{2^\ell},[[\mathbf{y}^{(1)}]]_{2^\ell})$
    \State \textbf{Output} $\langle \mathbf{y}^{\mathbb{C}}\rangle$
  \end{algorithmic}
\end{algorithm}

In an alternative pipeline, we evaluate the mid-range polynomial branch entirely in MPC after a single complex CKKS-to-MPC conversion and then convert the final outputs back to CKKS, as shown in Algorithm~\ref{alg:secure-gelu-bolt-mpcpoly-complex}.

\begin{algorithm}[t]
  \caption{Secure GELU with no pre-evaluation}
  \label{alg:secure-gelu-bolt-mpcpoly-complex}
  \begin{algorithmic}[1]\raggedright
    \Require $P_1$ holds a complex CKKS ciphertext $\langle \mathbf{x}^{\mathbb{C}}\rangle$ at scale $\Delta$
    encrypting $\mathbf{x}^{\mathbb{C}}=\mathbf{x}^{(0)}+i\mathbf{x}^{(1)}$, with $\mathbf{x}^{(0)},\mathbf{x}^{(1)}\in\mathbb{R}^n$
    \Require Public $a,b,c,d,e$ and thresholds $-2.7,0,2.7$ in MPC fixed-point over $\mathbb{Z}_{2^\ell}$
    \Ensure $P_1$ outputs $\langle \mathbf{y}^{\mathbb{C}}\rangle$ with $\mathbf{y}^{\mathbb{C}}\approx \mathbf{y}^{(0)}+i\mathbf{y}^{(1)}$,
    $\mathbf{y}^{(j)}\approx \mathrm{ApproxGELU}(\mathbf{x}^{(j)})$

    \State \textbf{Convert:} $([[\mathbf{x}^{(0)}]]_{2^\ell},[[\mathbf{x}^{(1)}]]_{2^\ell}) \leftarrow
    \Pi_{\mathrm{C2M}}^{\mathbb{C}}(\langle \mathbf{x}^{\mathbb{C}}\rangle)$

    \State \textbf{MPC:} for $j\in\{0,1\}$ compute powers and candidates
    \For{$j\in\{0,1\}$}
    \State $[[\mathbf{x}_2^{(j)}]] \gets \Pi_{\times}([[\mathbf{x}^{(j)}]]_{2^\ell},[[\mathbf{x}^{(j)}]]_{2^\ell})$
    \State $[[\mathbf{x}_3^{(j)}]] \gets \Pi_{\times}([[\mathbf{x}_2^{(j)}]],[[\mathbf{x}^{(j)}]]_{2^\ell})$
    \State $[[\mathbf{x}_4^{(j)}]] \gets \Pi_{\times}([[\mathbf{x}_2^{(j)}]],[[\mathbf{x}_2^{(j)}]])$
    \State $[[\mathbf{f}_0^{(j)}]] \gets
    a[[\mathbf{x}_4^{(j)}]] - b[[\mathbf{x}_3^{(j)}]] + c[[\mathbf{x}_2^{(j)}]] + (0.5-d)[[\mathbf{x}^{(j)}]]_{2^\ell} + e$
    \State $[[\mathbf{f}_1^{(j)}]] \gets
    a[[\mathbf{x}_4^{(j)}]] + b[[\mathbf{x}_3^{(j)}]] + c[[\mathbf{x}_2^{(j)}]] + (0.5+d)[[\mathbf{x}^{(j)}]]_{2^\ell} + e$
    \EndFor

    \State \textbf{MPC:} for $j\in\{0,1\}$ set
    $[[\mathbf{b}_{-}^{(j)}]]\!\gets\!\Pi_{\mathrm{cmp}}([[\mathbf{x}^{(j)}]]_{2^\ell},-2.7)$,
    $[[\mathbf{b}_{0}^{(j)}]]\!\gets\!\Pi_{\mathrm{cmp}}([[\mathbf{x}^{(j)}]]_{2^\ell},0)$,
    $[[\mathbf{b}_{+}^{(j)}]]\!\gets\!\Pi_{\mathrm{cmp}}([[\mathbf{x}^{(j)}]]_{2^\ell},2.7)$
    \State \textbf{MPC:} for $j\in\{0,1\}$ derive indicators
    $[[\mathbf{z}_0^{(j)}]] \gets [[\mathbf{b}_{-}^{(j)}]]\oplus[[\mathbf{b}_{0}^{(j)}]]$,
    $[[\mathbf{z}_1^{(j)}]] \gets [[\mathbf{b}_{0}^{(j)}]]\oplus[[\mathbf{b}_{+}^{(j)}]]$,
    $[[\mathbf{z}_2^{(j)}]] \gets 1 \oplus [[\mathbf{b}_{+}^{(j)}]]$
    \State \textbf{MPC:} for $j\in\{0,1\}$ output
    $[[\mathbf{y}^{(j)}]] \gets
    \Pi_{\mathrm{mux}}([[\mathbf{f}_0^{(j)}]],[[\mathbf{z}_0^{(j)}]])
    +\Pi_{\mathrm{mux}}([[\mathbf{f}_1^{(j)}]],[[\mathbf{z}_1^{(j)}]])
    +\Pi_{\mathrm{mux}}([[\mathbf{x}^{(j)}]]_{2^\ell},[[\mathbf{z}_2^{(j)}]])$

    \State \textbf{Convert:} $\langle \mathbf{y}^{\mathbb{C}}\rangle \leftarrow
    \Pi_{\mathrm{M2C}}^{\mathbb{C}}([[\mathbf{y}^{(0)}]]_{2^\ell},[[\mathbf{y}^{(1)}]]_{2^\ell})$
    \State \textbf{Output} $\langle \mathbf{y}^{\mathbb{C}}\rangle$
  \end{algorithmic}
\end{algorithm}

\section{Ring and Field Conversion}
\label{app:ring-field}

We describe two-party primitives for converting additive shares between an odd-modulus ring $\mathbb{Z}_q$ and a power-of-two ring $\mathbb{Z}_{2^\ell}$, following the extension-based approach of BOLT, SiRNN, and BLB~\cite{pang2024bolt,rathee2021sirnn,xu2025blb}.
In \system{}, $q$ is the CKKS ciphertext modulus at the conversion boundary, i.e., $q=Q_{L_{\mathrm{conv}}}$.
For additive sharing over $\mathbb{Z}_M$, we write $[[u]]_M$ where $u_0+u_1\equiv u\pmod M$.
We interpret ring elements with signed fixed-point semantics via the centered representative
\(\mathrm{cl}_M(u)\in(-M/2,M/2]\), defined as the unique integer congruent to \(u \bmod M\) that lies in this interval.
We use a signed extension protocol $\Pi_{\mathrm{Ext}}$ that lifts shares to a larger power-of-two ring while preserving $\mathrm{cl}_M(\cdot)$ except with probability at most $2^{-\sigma}$.

\subsection{Field-to-ring conversion}
$\Pi_{\mathrm{Field2Ring}}$ converts $[[x]]_q$ to $[[x]]_{2^\ell}$.
Invoke $\Pi^{q,2^{\ell'}}_{\mathrm{Ext}}$ with $\ell' \ge \max(\lceil \log_2 q\rceil,\ell)$ to obtain shares in $\mathbb{Z}_{2^{\ell'}}$ that represent $\mathrm{cl}_q(x)$ (with error $\le 2^{-\sigma}$), then each party locally reduces its share modulo $2^\ell$.
We choose $\ell$ so that all valid values satisfy $|\mathrm{cl}_q(x)|<2^{\ell-1}$.

\subsection{Ring-to-field conversion}
$\Pi_{\mathrm{Ring2Field}}$ converts $[[m]]_{2^\ell}$ to $[[m]]_q$.
Invoke $\Pi_{\mathrm{Ext}}$ to lift to $\mathbb{Z}_{2^{\ell+\sigma}}$, obtaining representatives
$m'_0,m'_1\in[0,2^{\ell+\sigma})$ such that, except with probability $2^{-\sigma}$, the \emph{integer} sum satisfies
\[
  m'_0+m'_1 = 2^{\ell+\sigma}+\mathrm{cl}_{2^\ell}(m).
\]
(The $2^{\ell+\sigma}$ offset encodes the signed value as a non-negative representative in $\mathbb{Z}_{2^{\ell+\sigma}}$.)
Each party then maps locally into $\mathbb{Z}_q$ by
\[
  [[m]]^{(q)}_0 = m'_0 \bmod q,\qquad
  [[m]]^{(q)}_1 = (m'_1-2^{\ell+\sigma}) \bmod q,
\]
so that $[[m]]^{(q)}_0+[[m]]^{(q)}_1 \equiv \mathrm{cl}_{2^\ell}(m)\pmod q$ except with probability $2^{-\sigma}$.

\subsection{Security}
Both conversions consist of one invocation of $\Pi_{\mathrm{Ext}}$ followed by local computations.
If $\Pi_{\mathrm{Ext}}$ is simulation-secure, then the conversion view can be simulated from the input/output shares, leaking nothing beyond public parameters, up to statistical error $2^{-\sigma}$.

\section{Cost Analysis Model Discussion}
\label{app:costmodel}

We use the latency decomposition in Eq.~(3) and the decision rule in Eq.~(4) to compare boundary choices via $\Delta T_{\text{conv}}$ and $\Delta T_{\text{comp}}$. A boundary change can trade off two effects: it can reduce computation time by moving operators to the cheaper domain, and it can increase or decrease the conversion overhead at domain boundaries. These effects are captured by $\Delta T_{\text{comp}}$ and $\Delta T_{\text{conv}}$, respectively.

The minimal-conversion baseline $B$ defined in Section~5.1 serves as a normalization reference for boundary payload under a fixed tensor shape. It does not claim to minimize end-to-end latency and is not an optimizer by itself. We use Eq.~(4) as a screening tool to estimate the sign of the latency change when moving from $B$ to a candidate design $D$.

To illustrate how the cost model can guide boundary placement, we apply it to the GELU boundary decision under PhantomFHE across 12 instances: 3 models $\times$ 4 network profiles. The model estimates ``Expand'' for all BERT-large profiles and for WAN2/WAN3 on BERT-base and GPT2-base, and ``Minimal'' for LAN/WAN1 on the smaller models --- a non-trivial 8/4 split that illustrates the model can surface trade-offs varying across deployment settings. We note that $\Delta T_{\text{ckks}}$ and $C_{\text{round}}$ are measured under specific backend and network conditions, so the model's estimates should be treated as directional rather than precise.

In practice, $T_{\text{mpc}}$ can also vary across network profiles because interactive MPC primitives are communication-bound. As bandwidth decreases or round-trip latency increases, the pipeline can shift from compute-dominated to message-dominated behavior, which changes the relative weight of $\Delta T_{\text{conv}}$ and $\Delta T_{\text{comp}}$. Since the deployment network may be unknown at design time, we recommend evaluating candidate boundary placements across multiple network profiles before committing to a design.

\section{MPC Round Count Comparison}
\label{app:mpc_round_comparison}

Table~\ref{tab:mpc_round_comparison} compares interactive MPC rounds per layer. FHE-only systems (THOR, Powerformer) use zero MPC rounds but require bootstrapping. The BOLT and BLB counts are analytically derived from their published protocol descriptions. BOLT's softmax uses 6 rounds (tree max-finding for $L{=}64$). BLB adopts BumbleBee's softmax protocol~\cite{lu2025bumblebee}, which requires 18 rounds: $\lceil\log_2 L\rceil$ rounds for tree max-finding ($L{=}128$), 6 sequential squarings for negExp, plus comparison, mux, reciprocal, and scaling. \system{} uses MBMax (3 rounds) and BOLT-style GELU (4 rounds); MBLN requires zero interactive rounds because rowsum, broadcast, and affine scaling are all local operations on additive shares.

\begin{table}[t]
\centering
\footnotesize
\setlength{\tabcolsep}{4.5pt}
\renewcommand{\arraystretch}{1.15}
\caption{Interactive MPC rounds per layer. BOLT and BLB counts are analytically derived from published protocols. BS = bootstrapping invocations.}
\begin{tabular}{@{}l|cccc|cc@{}}
\toprule
 & Softmax & LN1 & GELU & LN2 & Total & BS \\
\midrule
THOR        & 0 & 0 & 0 & 0 & 0 & 8 \\
Powerformer & 0 & 0 & 0 & 0 & 0 & 3 \\
BOLT        & 6 & 4 & 4 & 4 & 18 & 0 \\
BLB         & 18 & 4 & 4 & 4 & 30 & 0 \\
\system{}   & 3 & 0 & 4 & 0 & 7 & 0 \\
\bottomrule
\end{tabular}
\label{tab:mpc_round_comparison}
\end{table}

\section{Proof of Proposition~1}
\label{app:scp-proof}

\begin{proof}
\textbf{(i) Zero remap.}
SCP matches all three direct FHE--FHE edges (QKV$\to$Score, QKV$\to$Value, Value$\to$OUT) via weight pre-permutation, and the three MPC-boundary edges satisfy Rules~2 and~3 with minimal packing.
Any alternative assignment that changes a producer--consumer matching on an FHE--FHE edge necessarily introduces $\mathrm{out}(K_i)\neq\mathrm{in}(K_j)$, since no other packing in~$\mathcal{F}$ satisfies the consumer kernel without repacking.

\textbf{$\Omega(m\log m)$ rotation lower bound.}
A remap between two distinct packings in $\mathcal{F}$ requires a general slot permutation on each ciphertext block. In the segment-column layout with $m$ rows, the permutation acts on $m$ slot groups. By the Halevi--Shoup decomposition~\cite{halevi2014algorithms}, any such permutation requires at least $\log_2 m$ Galois automorphisms (rotations), each involving a key-switching operation. With $B$ ciphertext blocks affected per edge, the total cost is $B \cdot \log_2 m$ rotations per edge, giving $\Omega(m\log m)$ across the $\Theta(m)$ affected slots per block.

\textbf{(ii) Minimal boundary count.}
At each FHE--MPC boundary, Rules~2 and~3 require that ciphertexts use minimal packing with $K_{\min}$ ciphertexts. Any expanded packing uses $K > K_{\min}$ ciphertexts, each of which must pass through the conversion protocol, increasing boundary communication proportionally.

\textbf{Quantitative example.}
For BERT-base ($m{=}128$, $C{=}128$, $\log_2 m{=}7$), disabling SCP adds 168 extra rotations per layer across the three FHE--FHE edges ($24$ blocks $\times$ $7$ rotations), or ${\sim}0.39$\,s per layer (${\sim}9\%$ of per-layer FHE time).
\end{proof}

\section{SCP Edge Analysis for All Models}
\label{app:scp-edges}

Table~\ref{tab:scp_edge_all_models} extends the BERT-base SCP analysis from the main text to all three evaluated models. Without SCP, each direct FHE--FHE edge requires $\log_2 m$ rotations per ciphertext block for slot permutation. MPC-boundary edges incur no FHE remapping because MPC resets the data format.

\begin{table}[t]
\centering
\scriptsize
\setlength{\tabcolsep}{4pt}
\renewcommand{\arraystretch}{1.15}
\caption{SCP remap analysis for all models ($t_{\mathrm{rot}}=2.32$\,ms).}
\label{tab:scp_edge_all_models}
\begin{tabular}{@{}l|rrr@{}}
\toprule
 & BERT-b & BERT-l & GPT2-b \\
\midrule
QKV(Q,K)$\to$Score & 84 & 56 & 36 \\
QKV(V)$\to$Value   & 42 & 28 & 18 \\
Value$\to$OUT      & 42 & 28 & 18 \\
\midrule
\#rot/layer        & 168 & 112 & 72 \\
\#rot $\times$ $t_{\mathrm{rot}}$ (ms) & 390 & 260 & 167 \\
Layers             & 12 & 24 & 12 \\
Predicted (s)      & 4.68 & 6.24 & 2.00 \\
Observed (s)       & ${\sim}$12 & ${\sim}$48 & ${\sim}$12 \\
Obs./Pred.         & 2.6$\times$ & 7.7$\times$ & 6.0$\times$ \\
\bottomrule
\end{tabular}
\end{table}

The observed overhead is 2.6--7.7$\times$ the rotation-only prediction because the actual slot permutation also involves plaintext-multiply masking and ciphertext addition for recombination, which are $O(1)$ per rotation but not counted in the rotation proxy.

\subsection{Repacking Algorithm in the \texttt{w/o~SCP} Ablation}
\label{app:repack-algo}

\noindent\textbf{Halevi--Shoup slot-permutation decomposition.}
Applying an arbitrary permutation to the slots of a CKKS ciphertext cannot be done in a single operation.
Halevi and Shoup~\cite{halevi2014algorithms} show that any target slot permutation $\pi$ can be expressed as a linear combination of rotations:
\[
\mathrm{Repack}(\langle\mathbf{v}\rangle)
= \sum_{k=0}^{\log_2 m - 1}
  \mathrm{Rot}\!\left(\langle\mathbf{v}\rangle,\, 2^k\right)
  \cdot \mathbf{m}_k,
\]
where each $\mathbf{m}_k$ is a plaintext binary mask that selects the slots whose target position is reached by a rotation of $2^k$.
This decomposes the permutation into $\log_2 m$ rotation--mask--accumulate (RMA) steps.
Each RMA step requires one key-switched rotation, one plaintext multiplication, and one ciphertext addition, giving $\log_2 m$ rotations, $\log_2 m$ plaintext multiplications, and $\log_2 m$ additions per ciphertext block.

In the \texttt{w/o SCP} ablation (Table~11), we insert a generic slot-permutation repack at every direct FHE--FHE edge. The repacking converts one packing format to another by cycling through $\log_2 m$ RMA steps as described above.

Concretely, the ablation inserts repacking at the three direct FHE--FHE edges per layer identified in Section~3.1.1: (1)~QKV$\to$Score remaps Q and K from segment-column to the score-friendly layout ($2 \times \lceil d/C\rceil$ blocks, $\log_2 m$ rotations each); (2)~QKV$\to$Value remaps V from segment-column to head-major order ($\lceil d/C\rceil$ blocks); and (3)~Value$\to$OUT remaps the value output from head-major to segment-column ($\lceil d/C\rceil$ blocks). Edges that cross an MPC boundary (e.g.\ Score$\to$Softmax, OUT$\to$LN) are not affected because MPC resets the data format via decryption and re-encryption.
The total rotation count per layer is the sum reported in Table~\ref{tab:scp_edge_all_models}. The observed latency gap exceeds the rotation-only estimate because each rotation step is accompanied by a masking multiplication and an accumulation addition, and the large number of sequential CKKS operations incurs additional GPU kernel-launch and memory-transfer overhead.

\section{Cost Model Validation Details}
\label{app:costmodel-validation}

\subsection{GELU Boundary Predictions}

Table~\ref{tab:cost_model_summary} summarizes the cost model's boundary decision across both backends. PhantomFHE produces a non-trivial 8/4 split, while Liberate.FHE predicts Minimal for all 12 cases due to its $18\times$ higher CKKS compute time. Table~\ref{tab:gelu_validation_full} provides the full per-instance breakdown for PhantomFHE.

\begin{table}[t]
\centering
\footnotesize
\caption{Cost model boundary decision summary.}
\label{tab:cost_model_summary}
\resizebox{\columnwidth}{!}{%
\begin{tabular}{@{}l|cccc|cccc@{}}
\toprule
 & \multicolumn{4}{c|}{PhantomFHE} & \multicolumn{4}{c}{Liberate.FHE} \\
Model & LAN & WAN1 & WAN2 & WAN3 & LAN & WAN1 & WAN2 & WAN3 \\
\midrule
BERT-base  & Min & Min & Exp & Exp & Min & Min & Min & Min \\
BERT-large & Exp & Exp & Exp & Exp & Min & Min & Min & Min \\
GPT2-base  & Min & Min & Exp & Exp & Min & Min & Min & Min \\
\bottomrule
\end{tabular}%
}
\end{table}

Table~\ref{tab:gelu_validation_full} reports the full prediction-vs-observation validation for the GELU boundary decision under PhantomFHE. The model estimates ``Expand'' for 8 instances and ``Minimal'' for 4, namely LAN/WAN1 on BERT-base and GPT2-base, and the estimates agree with the observed latency direction in all cases. The tightest positive margin is on WAN2 GPT2-base ($\Delta T = -42$\,ms), while the tightest negative margin is on WAN1 GPT2-base ($\Delta T = +52$\,ms).

\begin{table}[t]
\centering
\footnotesize
\setlength{\tabcolsep}{3.5pt}
\renewcommand{\arraystretch}{1.1}
\caption{GELU boundary prediction under PhantomFHE ($t_{\mathrm{rot}}=2.32$\,ms, $t_{\mathrm{ctmul}}=2.65$\,ms). 8 instances predict ``Expand,'' 4 predict ``Minimal.'' $K_{\text{extra}}$ is the additional complex ciphertexts for $F_0,F_1$.}
\label{tab:gelu_validation_full}
\begin{tabular}{@{}l|l|r|rrr|l@{}}
\toprule
Model & Profile & $K_{\text{extra}}$ & $\Delta T_{\text{conv}}$ & $\Delta T_{\text{comp}}$ & $\Delta T$ & Decision \\
\midrule
BERT-base & LAN  & 24 & $+$0.20\,s & $-$0.04\,s & $+$0.16\,s & Minimal \\
          & WAN1 & 24 & $+$0.50\,s & $-$0.39\,s & $+$0.12\,s & Minimal \\
          & WAN2 & 24 & $+$2.01\,s & $-$2.09\,s & $-$0.07\,s & Expand \\
          & WAN3 & 24 & $+$2.01\,s & $-$2.31\,s & $-$0.30\,s & Expand \\
\midrule
BERT-large & LAN  & 16 & $+$0.13\,s & $-$0.18\,s & $-$0.04\,s & Expand \\
           & WAN1 & 16 & $+$0.34\,s & $-$0.64\,s & $-$0.30\,s & Expand \\
           & WAN2 & 16 & $+$1.34\,s & $-$2.90\,s & $-$1.56\,s & Expand \\
           & WAN3 & 16 & $+$1.34\,s & $-$3.13\,s & $-$1.79\,s & Expand \\
\midrule
GPT2-base & LAN  & 12 & $+$0.10\,s & $-$0.02\,s & $+$0.08\,s & Minimal \\
          & WAN1 & 12 & $+$0.25\,s & $-$0.20\,s & $+$0.05\,s & Minimal \\
          & WAN2 & 12 & $+$1.01\,s & $-$1.05\,s & $-$0.04\,s & Expand \\
          & WAN3 & 12 & $+$1.01\,s & $-$1.28\,s & $-$0.27\,s & Expand \\
\bottomrule
\end{tabular}
\end{table}

\subsection{Liberate.FHE Decision Boundary}

The cost model includes the CKKS pre-evaluation cost of the GELU polynomial candidates $F_0$ and $F_1$ as a negative $\Delta_{\mathrm{ctmul}}$ term in Eq.~Eq.~(6). Each degree-5 polynomial evaluation requires 3 ctmuls per complex ciphertext block, giving $6 \lceil B_{\mathrm{ff}}/2\rceil$ extra ctmuls total (e.g., 72 for BERT-base). Under PhantomFHE ($t_{\mathrm{ctmul}}=2.65$\,ms), this adds only $\sim$0.19\,s and is outweighed by MPC round savings on higher-latency profiles. Under Liberate.FHE ($t_{\mathrm{ctmul}}=48.66$\,ms, $18\times$ PhantomFHE), the same 72 ctmuls cost $\sim$3.5\,s, far exceeding MPC savings, so all 12 Liberate instances correctly predict ``Minimal.'' This backend-dependent discrimination validates that the model captures a real trade-off rather than a degenerate one.

\section{Total Communication Including Conversion Payload}
\label{app:total_comm}

\noindent\textbf{Cross-paper comparison note.}
We follow the cross-paper comparison methodology commonly used in this area: the baseline metrics reported throughout this paper are taken from the corresponding publications under their stated experimental settings, and our claims against prior hybrid systems should be interpreted in that sense rather than as a controlled re-implementation on a single platform. Because each prior system uses its own CKKS library, GPU hardware, and MPC implementation, a fully matched head-to-head rerun is often infeasible. Likewise, our WAN latency figures are obtained by applying an analytical network overlay to measured communication volume and MPC round counts, rather than by cross-site deployment; this is the same evaluation methodology adopted by prior hybrid Transformer systems such as BOLT, BumbleBee, and BLB. Section~7.3 therefore complements the main comparison with a matched-backend study that isolates library effects.

Table~\ref{tab:total_comm} estimates total communication (online MPC traffic plus FHE--MPC conversion payload) for BERT-base. The conversion payload per layer is computed from the ciphertext counts in Table~4 at 256\,KB per ciphertext ($n=16384$, $\sim$15 RNS primes). For all systems, conversion payload is small relative to online MPC traffic (0.2--0.5\,GB vs.\ 2--64\,GB), so the relative ranking is preserved under total communication.

\begin{table}[ht]
\centering
\footnotesize
\setlength{\tabcolsep}{5pt}
\caption{Estimated total communication for BERT-base (12 layers).}
\label{tab:total_comm}
\begin{tabular}{@{}l|rrr@{}}
\toprule
Framework & Online (GB) & Conv.\ payload (GB) & Total (GB) \\
\midrule
BOLT       & 63.6 & 0.28 & 63.9 \\
BumbleBee  &  5.8 & 0.28 &  6.1 \\
BLB        &  3.0 & 0.49 &  3.5 \\
\system{}  &  2.2 & 0.21 &  \textbf{2.4} \\
\bottomrule
\end{tabular}
\end{table}

\section{Additional Generality Examples}
\label{app:generality}

The SCP packing rules and the boundary cost model are stated in terms of kernel composition rather than Transformer-specific semantics. To illustrate this, we apply the same reasoning to two analytical examples.

\subsection{2-Layer CNN}

Consider a CNN with Conv1 ($3\times3$, 16 filters) $\to$ ReLU $\to$ Conv2 ($3\times3$, 32 filters) $\to$ ReLU $\to$ FC, operating on $28\times28$ input images with $n=8192$ CKKS slots. Conv1 uses 144 rotations ($9 \times 16$), Conv2 uses 288 rotations ($9 \times 32$), and the FC layer uses 10 rotations. Without co-designed packing, the Conv1$\to$ReLU and Conv2$\to$ReLU boundaries each require repacking (6 ciphertext blocks total), costing $\sim$14\,ms under PhantomFHE primitive latencies.

\subsection{3-Layer MLP}

Consider an MLP with FC1 ($768\to3072$) $\to$ GELU $\to$ FC2 ($3072\to3072$) $\to$ GELU $\to$ FC3 ($3072\to768$), using $n=16384$ slots and $m=128$, matching the BERT-base hidden and FFN dimensions. With $C=128$ segments per ciphertext, FC1 and FC2 each produce $\lceil 3072/128\rceil=24$ output blocks. Without co-designed packing, FC1$\to$GELU and FC2$\to$GELU require repacking 48 blocks at $\log_2(128)=7$ rotations each, totaling 336 extra rotations ($\sim$111\,ms). With the GELU boundary cost model, each GELU block also admits an ``Expand'' option analogous to the Transformer GELU pre-evaluation, saving 3 MPC rounds per block at the cost of $2\times8=16$ extra complex ciphertexts per GELU.

These examples are analytical illustrations; full empirical evaluation on non-Transformer workloads is left to future work.

\section{Powerformer-Style Distillation}
\label{app:powerformer_training}

For each GLUE task, we fine-tune a plaintext BERT-base teacher with standard nonlinearities following Powerformer~\cite{park2025powerformer}. We then distill an HE-compatible student that keeps the same macro-architecture but replaces all nonlinear components with the approximations used in our protocol. Distillation helps the student adapt to the changed activation and normalization behavior introduced by these approximations while preserving task accuracy.

Let the teacher be $f_t(x;W_t)$ and the student be $f_s(x;W_s,\theta)$. For input~$x$ with label~$y$ and logits $z_t,z_s$, we minimize
\begin{align*}
\mathcal{L}
&=
\lambda_{\mathrm{sup}}\,\mathcal{L}_{\mathrm{CE}}(z_s,y)
+\lambda_{\mathrm{kd}}\,\mathcal{L}_{\mathrm{KD}}(z_s,z_t)\\
&\quad+\lambda_{\mathrm{int}}
\sum_{\ell\in\mathcal{S}}
\left\lVert h^{(\ell)}_s - h^{(\ell)}_t \right\rVert_2^2 .
\end{align*}
where \(h^{(\ell)}_s,h^{(\ell)}_t\) are matched intermediate tensors and \(\mathcal{S}\) is the set of layers used. The logit distillation term uses temperature \(T\)
\begin{align*}
  \mathcal{L}_{\mathrm{KD}}(z_s,z_t)
  =
  T^2\ \mathrm{KL}\!\left(
    \mathrm{softmax}\!\left(\frac{z_t}{T}\right)
    \ \bigg\|\
    \mathrm{softmax}\!\left(\frac{z_s}{T}\right)
  \right).
  \label{eq:distill_kd}
\end{align*}

In \system{}, we learn the Softmax and LayerNorm parameters \(\theta\) jointly with student weights \(W_s\) during distillation, and deploy \((W_s,\theta)\) unchanged for encrypted inference. We follow Powerformer-style hyperparameters with batch size 64 and learning rate \(5\times10^{-5}\). We also use early stopping with an extended distillation budget and seeds \(\{0,42,777\}\).

\end{document}